\documentclass[lineno]{biometrika}

\usepackage{amsmath}
\usepackage{newtxtext}
\usepackage[subscriptcorrection]{newtxmath}
\usepackage[plain,noend]{algorithm2e}

\usepackage[hyperindex]{hyperref}
\usepackage{url}

\newcommand{\Nor}{\mathcal{N}}
\DeclareMathOperator*{\argmax}{arg\,max} % Jan Hlavacek

\graphicspath{{./art/}}

\makeatletter
\renewcommand{\algocf@captiontext}[2]{#1\algocf@typo. \AlCapFnt{}#2} % text of caption
\def\@algocf@capt@plain{top}
\renewcommand{\algocf@makecaption}[2]{%
  \addtolength{\hsize}{\algomargin}%
  \sbox\@tempboxa{\algocf@captiontext{#1}{#2}}%
  \ifdim\wd\@tempboxa >\hsize%     % if caption is longer than a line
    \hskip .5\algomargin%
    \parbox[t]{\hsize}{\algocf@captiontext{#1}{#2}}% then caption is not centered
  \else%
    \global\@minipagefalse%
    \hbox to\hsize{\box\@tempboxa}% else caption is centered
  \fi%
  \addtolength{\hsize}{-\algomargin}%
}
\makeatother

\begin{document}

\jname{Biometrika}
%% The year, volume, and number are determined on publication
\jyear{2024}
\jvol{103}
\jnum{1}
\cyear{2024}
%% The \doi{...} and \accessdate commands are used by the production team
%\doi{10.1093/biomet/asm023}
\accessdate{Advance Access publication on 29 November 2024}

%% These dates are usually set by the production team
\received{2 January 2017}
\revised{1 August 2023}

%% The left and right page headers are defined here:
\markboth{M. Goplerud et~al.}{Variational mixed models}

%% Here are the title, author names and addresses

\title{Partially factorized variational inference for high-dimensional mixed models}

\author{M. Goplerud}
\affil{Department of Government, University of Texas at Austin
\email{mgoplerud@austin.utexas.edu}}

\author{O. Papaspiliopoulos}
\affil{Bocconi University \email{omiros@unibocconi.it} }

\author{\and G. Zanella}
\affil{Department of Decision Sciences, Bocconi University, via Roentgen 1, 20136
Milan, Italy\email{giacomo.zanella@unibocconi.it}}

\nolinenumbers
\maketitle

\begin{abstract}
While generalized linear mixed models  are a fundamental tool in applied statistics, many specifications, such as those involving categorical factors with many levels or interaction terms, can be computationally challenging to estimate due to the need to compute or approximate high-dimensional integrals.
Variational inference is a popular way to perform such computations, especially in the Bayesian context. 
However, naive use of such methods can provide unreliable uncertainty quantification.
We show that this is indeed the case for mixed models, proving that standard mean-field variational inference dramatically underestimates posterior uncertainty in high-dimensions.
We then show how appropriately relaxing the mean-field assumption leads to methods whose uncertainty quantification does not deteriorate in high-dimensions, and  whose total computational cost scales linearly with the number of parameters and observations.
Our theoretical and numerical results focus on mixed models with Gaussian or binomial likelihoods, and rely on connections to random graph theory to obtain sharp high-dimensional asymptotic analysis.
We also provide generic results, which are of independent interest, relating the accuracy of variational inference to the convergence rate of the corresponding coordinate ascent   algorithm that is used to find it. 
Our proposed  methodology is implemented in the R package \href{https://github.com/mgoplerud/vglmer}{vglmer}. Numerical results with simulated and real data examples illustrate the favourable computation cost versus accuracy trade-off of our approach compared to various alternatives. 
\end{abstract}


\section{Introduction}
\label{sec:intro}

Generalized linear mixed models involve fixed and random effects associated with different sets of covariates (see, e.g., Chapters 2 and 3.4 of \citealt{wood}). They are a foundational and widely used tool across multiple disciplines. 
In many applications, categorical factors associated with the random effects have a large number of levels (e.g., in the hundreds or thousands), which leads to models with many parameters. 
When interaction terms are included as random effects, the number of parameters also becomes large very quickly. We refer to the resulting models as high-dimensional. %\emph{high-dimensional}.
 This setting is common in current social science applications; for example, political scientists use these models as part of estimating public opinion at hundreds or thousands of geographic units (e.g., \citealt{warshaw2012district,broockman2018bias}) or rely on models with so-called deep interactions %that involve many random effects 
(\citealt{ghitza}).  A different application of such models is for recommendation systems where the categorical factors correspond to customer and products (e.g., \citealt{perry,GaoOwen2017EJS}).

In this article, we focus on variational inference  for high-dimensional mixed models. 
The goal of variational inference is to approximate a probability distribution $\pi$, typically the posterior distribution of a Bayesian model,
by a more tractable one in a family $\mathcal{Q}$ of approximating distributions \citep{bishop}. 
Arguably the most popular methodology approximates $\pi$ with the distribution $q^*\in\mathcal{Q}$ that minimizes the Kullback-Leibler divergence to $\pi$. Very common is mean-field variational inference where  $\mathcal{Q}$ assumes independence across parameters and $q^*$ is computed with the so-called coordinate ascent variational inference algorithm (see, e.g., \citealt{blei2017variational}).
While being computationally convenient, the mean-field assumption can lead to a poor approximation, especially in terms of underestimation of uncertainty relative to $\pi$.
This is indeed the case for high-dimensional mixed models, as we formally show in Section \ref{sec:theory-GLMM} (see in particular Theorem \ref{th:ff}). 
There, under certain assumptions, we quantify the degree of variance underestimation of %the mean-field approximation 
$q^*$ relative to $\pi$ and show it
increases with the number of observations $n$ and parameters $p$. 
The results suggest  that standard mean-field approximations, which we refer to as fully factorized,  
do not provide reliable uncertainty quantification for high-dimensional mixed models.

Motivated by such results, we consider alternative families  of  variational approximations that differ in the dependence structure they assume on fixed and random effects. 
In particular, we introduce a partially factorized  family which leads to a favourable trade-off between computational cost and approximation accuracy. 
%The methodology is described in Sections \ref{sec:var-fam} and \ref{sec:methods}
We show that, under assumptions similar to the ones of Theorem \ref{th:ff} (see Theorems \ref{th:pf} and \ref{th:rg}),  both the degree of variance underestimation obtained with the new family, and the convergence rate of the coordinate ascent algorithm required to compute the variational approximation,   do not deteriorate as $n$ and $p$ increase. Equivalently, the number of iterations required to approximate $q^*$ with a fixed level of accuracy does not increase with $n$ and $p$. 
In terms of practical implications,  the results suggest that the improvement in accuracy comes with a negligible increase in computational cost.

The results are based on an interesting duality between the convergence rate of the coordinate ascent algorithm and the accuracy of the corresponding $q^*$. 
This connection, which is not specific to mixed models and is of independent interest, is obtained  in Section \ref{sec:theory} (see in particular Theorem \ref{th:rs-pf-cavi-conv}) for generic Gaussian targets. 
Intuitively, the result implies that when the variational approximation is poor in terms of uncertainty quantification, the corresponding coordinate ascent algorithm is also slow to converge and vice versa. %Our result formalises an empirical observation noted in previous research, e.g., \cite{ncp-vi}. 

We illustrate our proposed methodology and theory through numerical simulations and the deep interaction models from \cite{ghitza}. 
Despite dealing with more complex settings than the ones analysed in the theoretical part of this paper, 
the numerics confirm what the theory predicts and 
illustrate good computational cost versus approximation accuracy trade-offs when using our proposed partially factorized variational inference. All examples can be reproduced using code found at \url{https://github.com/mgoplerud/pfvi_glmm_replication}. 
The connections and differences between our approach and previous works that explore modified or relaxed versions of the mean-field assumption are discussed in  Section \ref{sec:lit}. 
%An appealing feature of our proposed of PF-VI methodology for GLMMs is that the corresponding CAVI algorithm can be defined in closed form and does not require stochastic optimization, resulting in a total cost of the overall methodology that scales only linearly with $n$ and $p$, while providing provably improved uncertainty quantification. 
%In terms of theory, our main contributions are the analysis of the proposed methodology for high-dimensional GLMMs and the general connection between accuracy and convergence speed.

\section{Generalized linear mixed models}
\label{sec:models}
We consider mixed models of the form
\begin{equation}
\label{eq:glmm}
\begin{aligned}
y_i &\mid \eta_i \sim p(y_i \mid \eta_i),\quad \eta_i = x_i^T\beta + \sum_{k=1}^K z_{i,k}^T \alpha_k,&i=1,\dots,n,\\
\alpha_{k,g} &\sim \Nor \left(0, \gamma \Sigma_k\right), \quad \alpha_k = (\alpha_{k,1}^T, \cdots, \alpha_{k,G_k}^T)^T,&g=1,\dots,G_k;k=1,\dots,K.
\end{aligned}
\end{equation}
Here $n$ denotes the number of observations, $\eta_i\in\mathbb{R}$ are linear predictors,  $\beta\in\mathbb{R}^{D_0}$ are commonly referred to as fixed effects (see, e.g., Chapter 9 of \citealt{diggle}) and the corresponding covariates $x_i\in\mathbb{R}^{D_0}$ contain an intercept and other variables of interest, such as the respondent's income in our analyses in Section~\ref{sec:numerics}. %Throughout, we interpret $d$ dimensional vectors as $d\times 1$ matrices. 
The model contains $K$ sets of random effects, $\alpha_k$ for $k \in \{1, \cdots, K\}$,  
corresponding to categorical factors with $G_k$ levels each, such as the state in which a survey respondent lives. 
%The associated parameters, $\bm{\alpha}_{k,g}\in\mathbb{R}^{D_k}$ for $g=1,\dots,G_k$, are assigned a normal prior and are stacked into a vector $\bm{\alpha}_{k}\in\mathbb{R}^{G_kD_k}$ with corresponding covariates $\bm{w}_{i,k}\in\mathbb{R}^{D_k}$. 
The associated parameters, $\alpha_{k,g}\in\mathbb{R}^{D_k}$ for $g=1,\dots,G_k$, have corresponding covariates $w_{i,k}\in\mathbb{R}^{D_k}$ and are assigned a normal prior.
A random intercept model has $D_k = 1$ and $w_{i,k} = 1$, while models with $D_k > 1$ are commonly referred to as ones with  random slopes. For each factor $k$, an observation $i$ is assigned to exactly one level whose membership is encoded in the one-hot vector $m_{i,k} = (m_{i,k,1}, \dots, m_{i,k,G_k}) \in \{0, 1\}^{G_k}$ with $\sum_{g=1}^{G_k}m_{i,k,g}=1$. 
We define the vector $\alpha_{k}=(\alpha_{k,1}^T,\dots,\alpha_{k,G_K}^T)^T\in\mathbb{R}^{G_kD_k}$ with covariates $z_{i,k} = m_{i,k} \otimes w_{i,k}\in\mathbb{R}^{G_kD_k}$, where $\otimes$ denotes the Kronecker product, so that 
%$\bm{z}_{i,k}^T \bm{\alpha}_k=\sum_{d=1}^{D_k}\sum_{g=1}^{G_k}\alpha_{k,g,d} m_{i,k,g}w_{i,k,d}$ or
$z_{i,k}^T \alpha_k=\sum_{g=1}^{G_k}m_{i,k,g}w_{i,k}^T\alpha_{k,g}
%=\sum_{d=1}^{D_k}\sum_{g=1}^{G_k}\alpha_{k,g,d} m_{i,k,g}w_{i,k,d}
$.

We focus on two important special cases of \eqref{eq:glmm}, which vary by the choice of likelihood function, $p(y_i | \eta_i)$, and discuss extensions in Section \ref{sec:conc}. First, we assume that it  is Gaussian with an identity link and, according to common practice and  software implementations (e.g., \citealt{bates2015lmer}; \citealt{wood}), we scale the variance of the random effects by the residual variance:
\begin{equation}
\begin{aligned}
%y_i \mid \eta_i &\sim 
p(y_i \mid  \eta_i)&=
\Nor \left(y_i; \eta_i, \sigma^2\right)\,, \quad \gamma=\sigma^2\,.
\end{aligned}
\label{eq:gauss}
\end{equation}
Second, we consider a binomial model with a logistic link function where we use Polya-Gamma data augmentation \citep{polson2013polyagamma} that introduces additional latent variables $\omega_i$:
\begin{equation}
\begin{aligned}
% y_i \mid \eta_i &\sim
p(y_i \mid \eta_i)&= \mathrm{Binomial}\left(y_i; n_i, \frac{\exp(\eta_i)}{1+\exp(\eta_i)}\right)\,, \quad \gamma=1\,, \\
p(y_i, \omega_i \mid \eta_i) &= \frac{1}{2^{n_i}} \exp\left(\left[y_i - n_i/2\right]\eta_i - \omega_i \frac{\eta_i^2}{2}\right) f_{PG}(\omega_i | n_i, 0)\,;
\end{aligned}
\label{eq:binom}
\end{equation}
 where $f_{PG}(\omega| b,c)$ indicates a Polya-Gamma variable with parameters $b$ and $c$. Alternative treatments of the binomial model are discussed in Section \ref{sec:lit}.

 We adopt a Bayesian approach and assign prior distributions to the model parameters. In particular,  the components of $\beta$ are assigned improper flat priors and the matrices $\Sigma_k$ are given inverse Wishart distributions, $\mathrm{IW}(a_k^0,\Phi_k^0)$, which become simply inverse gamma distributions when  $D_k=1$. In the Gaussian model, we also assign the improper prior $1/\sigma^2$ to $\sigma^2$.

 Throughout the article we denote by  $\theta=(\theta_0, \theta_1,\dots,\theta_K)=(\beta, \alpha_1,\dots,\alpha_K)$ the vector of all fixed and random effects and by $y=(y_1,\dots,y_n)$ the $n$ data points. We let $\phi$ denote the collection of all $\Sigma_k$ and any additional parameters or latent variables, e.g., the residual variance $\sigma^2$ in the Gaussian model or the Polya-Gamma latent variables $(\omega_i)_{i=1}^n$ in the binomial model.

\section{Partially factorized variational inference}

\subsection{Mean-field approximations}
\label{sec:var-fam}

%We consider variational methods to approximate the posterior $p(\bm{\theta},\phi\midy)$.
%While we are mostly interested in situations where $p(\bm{\theta},\phi\midy)$ is the posterior of a GLMMs defined as in Section~\ref{sec:models}, the methodology discuss here applies to arbitrary $(\bm{\theta},\phi)=(\bm{\theta}_0,\ldots,\bm{\theta}_K,\phi)$.

We consider variational approximations $q^*(\theta, \phi)$ of $p(\theta,\phi\mid y)$ defined as 
\begin{equation}
\label{eq:ELBO}
q^*%(\bm{\theta}, \phi) 
= \argmax_{q \in \mathcal{Q}} \mathbb{E}_{q(\theta, \phi)}\left[\log p(\theta, \phi,y) - \log q(\theta, \phi)\right].
\end{equation}
%for different families of approximating densities $\mathcal{Q}$.
The expectation maximized in \eqref{eq:ELBO} is usually referred to as evidence lower bound (ELBO).
We note that while in our context $p(\theta,\phi\mid y)$ represents the joint posterior distribution of unknown parameters for GLMMs defined in Section~\ref{sec:models}, the notation and formulas of this section apply to arbitrary distributions defined over a set of variables $(\theta,\phi)$.

We consider approximations $q^*$ arising from different choices for the family of approximating densities $\mathcal{Q}$, 
which differ in the factorization they assume across the elements of $\theta$. 
First, we define the fully factorized
%\emph{fully factorized}  
family as
\begin{equation}
  \mathcal{Q}_{FF} = \left \{q\,:\, q(\theta,\phi) = \left[\prod_{k=0}^K q(\theta_k)\right] q\left(\phi\right)\right\}, 
  \label{eq:FFVI}
\end{equation}
where the above notation refers to the set of all joint probability distributions over $(\theta,\phi)$ that factorize as indicated. Second, we define the unfactorized
%\emph{unfactorized}
family as 
\begin{equation}
  \mathcal{Q}_{UF} = \left\{q\,:\, q(\theta,\phi) =  q(\theta) q(\phi)\right \}.
  \label{eq:UFVI}
\end{equation}
Finally, we define the partially factorized
%\emph{partially factorized} 
family as   
\begin{equation}
  \begin{aligned}
  \mathcal{Q}_{PF} & = \left\{q\,:\, q(\theta,\phi) =  q\left(\theta_{\mathcal{C}}|\theta_{\mathcal{U}}\right)
   \left[\prod_{k \in \mathcal{U}} q(\theta_{k})\right] q(\phi) \right \}\,,
  \end{aligned}
   \label{eq:PFVI}
\end{equation}
where $\theta_{\mathcal{C}}  =(\theta_{k})_{k \in \mathcal{C}}$, $\theta_{\mathcal{U}} =(\theta_{k})_{k \in \mathcal{U}}$ and $(\mathcal{C},\mathcal{U})$ is an arbitrary partition of $\{0,\dots,K\}$ in two blocks, i.e., $\mathcal{C}\subseteq \{0,\ldots, K\}$ and $\mathcal{U} = \{0, \ldots, K\} \setminus \mathcal{C}$. %\\
%  \mathcal{Q}_{PF} & = \left\{q\,:\, q(\bm{\theta},\phi) =  q\left(\bm{\theta}_{\mathcal{C}}|\bm{\theta}_{\mathcal{U}}\right)
%   \left[\prod_{k \in \mathcal{U}} q(\bm{\theta}_{k})\right] q(\phi) \right \}
Note that when $\mathcal{C} = \emptyset$ we obtain $\mathcal{Q}_{PF}=\mathcal{Q}_{FF}$ while when 
$\mathcal{C} = \{0,\ldots, K\}$ we obtain $\mathcal{Q}_{PF}=\mathcal{Q}_{UF}$. Later in the article, with some abuse of notation write $\mathcal{U}\backslash k$ to refer to the set $\mathcal{U}$ with $k\in \mathcal{U}$ excluded.

\subsection{Coordinate ascent variational inference}
\label{sec:pf_cavi}
The optimizer $q^*$ in  \eqref{eq:ELBO} is usually computed %approximated
 using a coordinate ascent algorithm that scans through the terms in the factorization defined by  $\mathcal{Q}$, maximizing the objective function in \eqref{eq:ELBO} with respect to one term while holding the other terms fixed.
For many classes of models and factorizations, the optimal densities in each coordinate  maximization step belong to specific parametric families and, as a result, the updating of factors boils down to updating their parameters.   
 Whereas the coordinate ascent steps for $\mathcal{Q}_{FF}$ are quite standard and can be found in textbooks \citep[ch. 10]{bishop}, those for $\mathcal{Q}_{PF}$ are less so and for this reason we characterize them in Proposition \ref{prop:PFCAVI} below.  The algorithm for finding $q^* \in \mathcal{Q}_{PF}$, is simply obtained by scanning sequentially through the coordinate updates in \eqref{eq:cavi_phi_general}-\eqref{eq:cavi_theta_U_general} below. 
 The proof of the Proposition is simple but instructive and it introduces a key intermediate distribution in \eqref{eq:tilde_pi_def} that our theory focuses on later. The specific densities involved for mixed models are given later in Section \ref{sec:cavi-glmm}. 

 \begin{proposition}
    Let $p(\theta,\phi\mid y)$ be the target density
and $\mathcal{Q}_{PF}$  be defined in \eqref{eq:PFVI}. 
    Then, the coordinate-wise optimizers of the objective function \eqref{eq:ELBO} are as follows:
    \begin{align}
      q(\phi) & \gets {1 \over c} \exp\left\{ \mathbb{E}_{ q(\theta)}[\log p(\theta,\phi, y)] \right\} 
\label{eq:cavi_phi_general}
\\
      q(\theta_{\mathcal{C}} \mid \theta_{\mathcal{U}}) & \gets 
{\exp\left \{\mathbb{E}_{q(\phi)} [\log p(\theta,\phi, y)] \right\}
        \over \int \exp\left \{\mathbb{E}_{q(\phi)} [\log p(\theta,\phi, y)] \right\} \mathrm{d} \theta_{\mathcal{C}}} 
\label{eq:cavi_theta_C_general}\\
      q(\theta_{k}) & \gets {1\over c}  
\exp\left\{
\mathbb{E}_{q(\theta_{\mathcal{U}\backslash k})} \left [\log \int 
\exp\left \{\mathbb{E}_{q(\phi)} [\log p(\theta,\phi, y)] \right\}
\mathrm{d} \theta_{\mathcal{C}} \right ]
\right\}
,\quad k \in \mathcal{U},
\label{eq:cavi_theta_U_general}
    \end{align}
where $c$ is a different constant of integration in each case, $q\left(\theta\right)=q\left(\theta_{\mathcal{C}}\mid \theta_{\mathcal{U}}\right)  \left[\prod_{k \in \mathcal{U}} q(\theta_{k})\right]$, $q(\theta_{\mathcal{U}\backslash k})  = \prod_{m\in \mathcal{U},m\neq k} q(\theta_{m})$, and where the updates in \eqref{eq:cavi_theta_U_general} assume that $q(\theta_{\mathcal{C}}\mid \theta_{\mathcal{U}})$ has already been updated according to \eqref{eq:cavi_theta_C_general}. 
    \label{prop:PFCAVI}
  \end{proposition}
  \begin{proof}
  The update in \eqref{eq:cavi_phi_general} is obtained by a standard argument for mean-field coordinate ascent, see e.g. Ch. 10 in \cite{bishop}. 
The ones in \eqref{eq:cavi_theta_C_general} and \eqref{eq:cavi_theta_U_general} are obtained as follows.
Let
  \begin{equation}\label{eq:tilde_pi_def}
 \pi(\theta) = {\exp\left \{\mathbb{E}_{q(\phi)} [ \log p(\theta,\phi, y)] \right\}
        \over \int \exp\left \{\mathbb{E}_{q(\phi)} [\log p(\theta,\phi, y)] \right\} \mathrm{d} \theta}
  \end{equation}
be a probability distribution over $\theta$. Note that  $\pi(\theta)$ depends on $q(\phi)$.
  Then, it holds

\[
\mathbb{E}_{q(\theta, \phi)}\left[\log p(\theta, \phi,y) - \log q(\theta, \phi)\right]
= \mathbb{E}_{q(\theta)} \left[\log \frac{ \pi(\theta_{\mathcal{C}} \mid \theta_{\mathcal{U}})}{q(\theta_{\mathcal{C}} \mid \theta_{\mathcal{U}}) }\right]
+\mathbb{E}_{q(\theta_{\mathcal{U}})} \left[\log \frac{\pi(\theta_{\mathcal{U}}) }{ q(\theta_{\mathcal{U}})} \right]+ c\,,
\]
where $c$ contains terms that do not depend on $q(\theta)$. Optimizing this over $q(\theta_{\mathcal{C}} \mid \theta_{\mathcal{U}})$ we obtain directly that the optimizer is  $\pi(\theta_{\mathcal{C}} \mid \theta_{\mathcal{U}})$, hence derive  \eqref{eq:cavi_theta_C_general}. 
Given $q(\theta_{\mathcal{C}} \mid \theta_{\mathcal{U}})=\pi(\theta_{\mathcal{C}} \mid \theta_{\mathcal{U}})$ and 
%Plugging the optimizer into the equation above, and exploiting 
$q(\theta_{\mathcal{U}})=\prod_{k\in \mathcal{U}} q(\theta_k)$, we obtain that the $q(\theta_k)$ coordinate-wise updates maximize
\begin{equation}
\mathbb{E}_{q(\theta{_{\mathcal{U}}})}\left[ \log \frac{\pi(\theta_{\mathcal{U}}) }{\prod_{k\in \mathcal{U}} q(\theta_k)}\right] + c,
  \label{eq:pfcavi-marg}
\end{equation}
where $c$ contains terms that do not depend on $q(\theta_{\mathcal{U}})$. The expression in \eqref{eq:cavi_theta_U_general} follows by the fact that the objective function in \eqref{eq:pfcavi-marg} coincides with the one of standard mean-field coordinate ascent applied to $\pi(\theta_{\mathcal{U}})$, and by the definition of $\pi(\theta)$.
  \end{proof}

It is useful to contrast the range of applicability of partially factorized coordinate ascent relative to the common fully factorized one. 
The update of $q(\phi)$ is  basically the same in both. 
That of $q(\theta_{\mathcal{C}} \mid \theta_{\mathcal{U}})$  is at least as easy with a partially factorized family compared to a fully factorized one, since in \eqref{eq:cavi_theta_C_general} we condition on $\theta_{\mathcal{U}}$ rather than integrating it out in the associated computations. 
The update of $q(\theta_{k})$ for $k \in \mathcal{U}$, instead, requires an additional integration (with respect to $\theta_{\mathcal{C}}$) relative to fully factorized variational inference. Section~\ref{sec:cavi-glmm} shows how to carry out this step efficiently for both Gaussian and binomial mixed models, since in these cases $\pi(\theta)$ is a high-dimensional Gaussian distribution (see, e.g., Proposition \ref{prop:pi}).
One should expect the computational cost of this step to grow with the size of $\mathcal{C}$, see  Section \ref{sec:comp_cost} for more details.  The connection between the partially factorized framework we develop here and other approaches to relaxing the mean-field assumption are discussed in Section \ref{sec:lit}.

\section{The accuracy-convergence duality for variational inference}
\label{sec:theory}

\subsection{Framework for theoretical analysis}

In this section, we derive an explicit connection between the accuracy of variational inference and the convergence  of the corresponding coordinate ascent algorithm. We focus on a framework with a generic Gaussian target distribution $\pi(\theta)$.  This is of interest as a canonical setting in its own right, but also as the intermediate target we encounter for mixed models when $q(\phi)$ is fixed; Section \ref{sec:theory-GLMM} uses these results to study accuracy and convergence for mixed models specifically.
Formally, this section considers the following set-up. 
\begin{assumption}
  \label{ass:sub-cavi-gauss-pi}
Let $\pi(\theta)\equiv \Nor(\theta; \mu,Q^{-1})$, with $\theta=(\theta_0,\dots,\theta_K)$ divided in $(K+1)$ blocks, and
\begin{equation}
\begin{aligned}\label{eq:fix-phi}
q^*
&= \argmax_{q \in \mathcal{Q}} V(q \Vert \pi)%\mathbb{E}_{q(\theta)}\left[\log\frac{ \pi(\theta)}{q(\theta)}\right]
,\quad
%which can be written as 
%\begin{equation}
%  q^*  = \argmax_{q \in \mathcal{Q}} V(q \Vert \pi),\quad
  V(q \Vert \pi)  = \mathbb{E}_{q(\theta)}\left[\log \pi(\theta)-\log q(\theta)\right]%\mathbb{E}_{q}[\log (\pi/q)]
% \label{eq:V}
%\end{equation}
,\\%\quad
 \mathcal{Q} &= \left\{q\;:\; q\left(\theta\right)=q\left(\theta_{\mathcal{C}}|\theta_{\mathcal{U}}\right)  \left[\prod_{k \in \mathcal{U}} q(\theta_{k})\right ]  \right \}\,,
\end{aligned}
\end{equation}
where $(\mathcal{C},\mathcal{U})$ is a partition of $\{0,\dots,K\}$ with $\mathcal{U}\neq \emptyset$.
\end{assumption}

When $\mathcal{C}=\emptyset$, the variational family coincides with the standard   mean-field one, therefore our results  also apply  and, 
to the best of our knowledge, are novel also in that specific case.

\subsection{Uncertainty quantification fraction}
We define the following metric, which measures the accuracy of a variational approximation $q$ for uncertainty quantification by comparing the variance of linear functions under $q$ to those under $\pi$.
We focus on variance underestimation, due to the broad understanding that mean-field variational inference is quite successful in approximating the mode of the target but it underestimates the variance \citep{blei2017variational}. Indeed, when $\pi$ is Gaussian, the mean of $q^*$ coincides with the one of $\pi$ (see, e.g., Chapter 10 in \citealt{bishop}).

\begin{definition}
  \label{def:uqf}
Let $\pi$ and $q$ be probability distributions over $\theta\in\mathbb{R}^p$. We define the uncertainty quantification fraction  between $\pi$ and $q$ as
%  For a target density $\pi(\theta)$ and an approximation $q(\theta)$ we define the uncertainty quantification fraction (UQF) between $\pi$ and $q$ as
    \begin{equation}\label{eq:def_uqf}
      UQF(q \Vert \pi) = \inf_{v \in \mathbb{R}^p\backslash\{0\}}
      {\mathrm{var}_{q(\theta)}(\theta^T v) \over \mathrm{var}_{\pi(\theta)}(\theta^T v)}. 
  \end{equation}
\end{definition}
The following propositions provide useful characterizations of the proposed measure. 

\begin{proposition}
  \label{prop:uqf}  
Let the covariance matrices of $\pi$ and $q$, %which we denote as
 $\mathrm{cov}_\pi(\theta)$ and $\mathrm{cov}_{q}(\theta)$, be invertible. Then
%  Consider a target density $\pi(\theta)$ and an approximation $q(\theta)$ with corresponding covariance matrices $\mathrm{cov}_\pi(\theta)$ and $\mathrm{cov}_{q}(\theta)$ assumed to be invertible. Then: 
  \begin{equation*}
    UQF(q \Vert \pi) =
    %\lambda_{min}(\tilde{\bm{\Sigma}} \bm{\Sigma}^{-1}) =
    \lambda_{min}(\mathrm{cov}_{q}(\theta)^{1/2} \mathrm{cov}_\pi(\theta)^{-1} \mathrm{cov}_{q}(\theta)^{1/2} ) =  1/ \lambda_{max}(\mathrm{cov}_\pi(\theta)\mathrm{cov}_{q}(\theta)^{-1})\,,
    \label{eq:uqf-eigen}
  \end{equation*}
  where $\lambda_{min}(\cdot), \lambda_{max}(\cdot)$ denote the minimum and the maximum eigenvalue. 
\end{proposition}
 \begin{proposition}
   \label{prop:uqf-bound-pf}
   Under Assumption \ref{ass:sub-cavi-gauss-pi} there is a unique   $q^*= \argmax_{q \in \mathcal{Q}} V(q\Vert \pi)$ and   $UQF(q^* \Vert \pi)  \leq 1$. 
 \end{proposition}
 Therefore, under Assumption \ref{ass:sub-cavi-gauss-pi},  larger values of $UQF(q^* \Vert \pi)$ indicate a better accuracy of the approximation.

\subsection{Coordinate ascent convergence rate and the uncertainty quantification fraction}\label{sec:duality}
We now relate $UQF(q^* \Vert \pi)$ to the convergence rate of the corresponding coordinate ascent algorithm. For mathematical convenience, we  consider a random scan version of coordinate ascent \citep{bhattacharya2023convergence}: At each iteration, the algorithm updates one term $q(\theta_{k})$ with $k\in\mathcal{U}$ chosen at random with uniform probabilities. Algorithm \ref{alg:rs-pf-cavi} in the Supplement provides a full description. The use of a random scan algorithm ensures that the convergence does not depend on the order in which the terms are updated, which makes theoretical results easier to interpret. The theorem is proved in the Supplement. 

\begin{theorem}
   \label{th:rs-pf-cavi-conv}
  Let $q_t$ be the density obtained after $t$ iterations of random scan coordinate ascent algorithm (i.e., Algorithm \ref{alg:rs-pf-cavi} in the Supplement). 
   Under Assumption \ref{ass:sub-cavi-gauss-pi}, for all starting $q_0$ and $t\geq 0$ we have
\begin{align}\label{eq:upper_bound}
\mathbb{E}[V(q^* \Vert \pi) - V(q_t \Vert \pi)] \leq \left(1 - {UQF(q^* \Vert \pi)\over |\mathcal{U}| } \right)^{|\mathcal{U}|t}  (V(q^* \Vert \pi) - V(q_0  \Vert \pi ))\,.
\end{align}
   On the other hand, there exist starting $q^\dagger_0$ for which
    \begin{align}\label{eq:lower_bound}
\mathbb{E}[V(q^* \Vert \pi) - V(q_t \Vert \pi)] \geq \left(1 - {UQF(q^* \Vert \pi)\over |\mathcal{U}|} \right)^{2|\mathcal{U}|t}  (V(q^*  \Vert \pi)-V(q^\dagger_0  \Vert \pi))\,,
   \end{align}
for all $t\geq 0$.
 The expectations in \eqref{eq:upper_bound} and \eqref{eq:lower_bound} are taken with respect to the random coordinates chosen at each iteration.  
 \end{theorem} 
Theorem \ref{th:rs-pf-cavi-conv} provides a direct connection between the exponential rate of convergence of random scan coordinate ascent and $UQF(q^* \Vert \pi)$. 
In particular, \eqref{eq:upper_bound} implies 
%$[V(q^* \Vert \pi) - V(\pi \Vert q_t)]
$$
\mathbb{E}[V(q^* \Vert \pi) - V(\pi \Vert q_t)]
%\left(1 - {UQF(q^* \Vert \pi)\over |\mathcal{U}| } \right)^{|\mathcal{U}|t}
\leq \exp\left(-UQF(q^* \Vert \pi)t\right)V(\pi \Vert q_0)\,,
$$
meaning that $q_t$ converges to $q^*$ in target values with exponential speed of $1/UQF(q^* \Vert \pi)$. In other words, the number of  iterations needed to converge scales as $1/UQF(q^* \Vert \pi)$. 
On the other hand, \eqref{eq:lower_bound} shows that the bound is tight modulo small constants. For example, \eqref{eq:lower_bound} implies that there exist starting $q^\dagger_0$ for which%\gzcom{(WRITE OUT 2.8)}
\begin{align*}
\mathbb{E}[V(q^* \Vert \pi) - V(\pi \Vert q_t)] \geq
\exp\left(-2.8\times UQF(q^* \Vert \pi)t\right)\,.
\end{align*}
%Thus, Theorem \ref{th:rs-pf-cavi-conv}
Therefore, under Assumption \ref{ass:sub-cavi-gauss-pi}, there is a direct and tight connection between the rate of convergence of coordinate ascent and the quality of the variational approximation. When the variational approximation is poor in terms of uncertainty quantification, the corresponding  algorithm is also slow to converge,  and vice versa. 
This is an interesting insight into coordinate ascent variational inference and shows that  the algorithms either slowly return a poor approximation or quickly recover a good one. 
This formalizes empirical observations noted in previous research (see, e.g., \cite{ncp-vi}). It also shows that studying $UQF(q^* \Vert \pi)$ is sufficient to characterize both approximation accuracy and algorithmic convergence speed, as we do in Section \ref{sec:theory-GLMM}. Note also that \cite{ncp-vi} discuss the connection between the convergence of coordinate ascent variational inference and that of the Gibbs sampler, and uses this to import the reparameterisation framework of \cite{ncp} within variational inference to accelerate convergence.

\section{Partially factorized variational inference for mixed models}
\label{sec:methods}

\subsection{Coordinate ascent updates}
\label{sec:cavi-glmm}

We work out the details of the generic formulae in \eqref{eq:cavi_phi_general}-\eqref{eq:cavi_theta_U_general} for the specific setting of mixed models with Gaussian and binomial likelihood. We do so for the partially factorized variational family since this includes as special cases the fully factorized and the unfactorized, that can also be derived from the results below.
We first write out explicitly the distribution $\pi(\theta)$ defined in \eqref{eq:tilde_pi_def}, as this plays a key role in both the convergence theory of Section \ref{sec:theory-GLMM} and the numerical experiments. 
In the following results we write $\mathrm{diag}[(v_j)_{j=1:d}]$ to denote a diagonal matrix with elements $v_j$ and $\mathrm{diag}[(U_\ell)_{\ell=1:L}]$ to denote a block diagonal matrix with diagonal blocks $U_\ell$. 
\begin{proposition}\label{prop:pi}
For model \eqref{eq:glmm}, likelihood as in \eqref{eq:gauss} or \eqref{eq:binom} and fixed $q(\phi)$, the distribution $\pi(\theta)$ in \eqref{eq:tilde_pi_def} is multivariate Gaussian with
\begin{align*}
\pi(\theta_{\mathcal{C}} | \theta_{\mathcal{U}}) &\equiv \Nor\left((W_{\mathcal{C}}^T W_{\mathcal{C}} + P_{\mathcal{C}})^{-1} W_{\mathcal{C}}^T (\nu-W_{\mathcal{U}} \theta_{\mathcal{U}}),(W_{\mathcal{C}}^T W_{\mathcal{C}} + P_{\mathcal{C}})^{-1}\right)  \\
  \pi(\theta_{\mathcal{U}}) & \equiv \Nor\left((P_{\mathcal{U}}+W_{\mathcal{U}}^T M_{\mathcal{C}}W_{\mathcal{U}})^{-1}W_{\mathcal{U}}^T M_{\mathcal{C}}\nu, (P_{\mathcal{U}}+W_{\mathcal{U}}^T M_{\mathcal{C}}W_{\mathcal{U}})^{-1}\right).
\end{align*}
Here $\nu=\left(\nu_1,\dots,\nu_n\right)\in\mathbb{R}^n$,
$D=\mathrm{diag}\left[\left(D_{ii}\right)_{i=1:n}\right]\in\mathbb{R}^{n\times n}$,  
$W_{\mathcal{C}}= \left[(W_{k})_{k \in \mathcal{C}}\right]\in\mathbb{R}^{n\times \sum_{k \in \mathcal{C}} D_kG_k}$, 
%$W_\mathcal{U} = \left[(W_{k})_{k \in \mathcal{U}}\right]$, 
$P_{\mathcal{C}} = \mathrm{diag}[ (P_{k})_{k \in \mathcal{C}}]\in\mathbb{R}^{\sum_{k \in \mathcal{C}}D_kG_k \times \sum_{k \in \mathcal{C}}D_kG_k}$ and analogously $W_{\mathcal{U}}$ and $P_{\mathcal{U}}$, with
\begin{align*}
M_{\mathcal{C}} & = I - W_{\mathcal{C}} (W_{\mathcal{C}}^T W_{\mathcal{C}} + P_{\mathcal{C}})^{-1} W_{\mathcal{C}}^T\in\mathbb{R}^{n\times n},\\
W_k &= D Z_k\in\mathbb{R}^{n\times D_kG_k},\,
P_k=I_{G_k} \otimes T_k\in\mathbb{R}^{D_kG_k\times D_kG_k}\quad\hbox{for }k>0\,,\\
%\hbox{ and }
W_0 &= D X\in\mathbb{R}^{n\times D_0},\,P_0 =O\in\mathbb{R}^{D_0\times D_0}\,,%\hbox{ for }k=0\,.
\end{align*}
where $O$ denotes a matrix of zeros. 
For the Gaussian likelihood as in \eqref{eq:gauss}
\begin{align}\label{eq:D_gauss}
\nu_i &= y_i \left(\mathbb{E}_{q(\sigma^2)}[1/\sigma^2]\right)^{1/2},\;
D_{ii}^2= \mathbb{E}_{q(\sigma^2)}(1/\sigma^2),\;
T_{k}= \mathbb{E}_{q(\sigma^2)}\left[1/\sigma^2\right]\mathbb{E}_{q(\Sigma_k)}[\Sigma_k^{-1}],
\end{align}
%\begin{align*}
%\nu_i &= y_i \sqrt{\mathbb{E}_{q(\sigma^2)}[1/\sigma^2]},\,D_{ii}^2= \mathbb{E}_{q(\sigma^2)}(1/\sigma^2)&i=1,\dots,n\\
%T_{k} &= \mathbb{E}_{q(\sigma^2)}\left[1/\sigma^2\right]\mathbb{E}_{q(\Sigma_k)}[\Sigma_k^{-1}]
%&k=1,\dots,K\,
%\end{align*}
while in the binomial likelihood as in \eqref{eq:binom}
\begin{align}\label{eq:D_binom}
\nu_i &= \left(y_i - n_i/2\right)/\left(\mathbb{E}_{q(\omega_i)}[\omega_i]\right)^{1/2},\;
D_{ii}^2=  \mathbb{E}_{q(\omega_i)}[\omega_i],\;
T_{k}= \mathbb{E}_{q(\Sigma_k)}[\Sigma_k^{-1}].
\end{align}
Above we use the convention $M_\emptyset = I\in\mathbb{R}^{n\times n}$  and 
 $\otimes$ denotes the Kronecker product.
\end{proposition}
The proof of Proposition \ref{prop:pi} is based on tedious but basic calculations and we omit it. 
As discussed in Proposition \ref{prop:PFCAVI} and its proof, the update of $q(\theta_{\mathcal{C}}\mid \theta_{\mathcal{U}})$ equals $\pi(\theta_{\mathcal{C}}\mid \theta_{\mathcal{U}})$, and the latter is obtained in Proposition \ref{prop:pi}. From \eqref{eq:cavi_theta_U_general} and  Proposition \ref{prop:pi}, the update of $q(\theta_k)$ for $k\in \mathcal{U}$ has the form $\Nor(\mu_k,\Lambda_k)$, where: 
\begin{align}
\label{eq:lambdaMultiply}\mu_k & = \Lambda_k W_k^TM_{\mathcal{C}} \left[\nu - \sum_{\ell \in \mathcal{U}, \ell\neq k} W_\ell \mu_\ell\right] \\
%\mu_k & =   \Lambda_k \left[(W_k M_{\mathcal{C}} - P_{\mathcal{C}} (W_{\mathcal{C}}^TW_{\mathcal{C}} + P_{\mathcal{C}})^{-1} W_{\mathcal{C}}^T) \nu\phantom{\sum_k} \right . \\
%   & \left . - (W_k  + P_{\mathcal{C}} (W_{\mathcal{C}}^TW_{\mathcal{C}} + P_{\mathcal{C}})^{-1} W_{\mathcal{C}}^T) \sum_{\ell \in \mathcal{U}, \ell \neq k} W_\ell \mu_\ell\right] \,,\\
\label{eq:lambdaWoodbury} \Lambda_k & = \Lambda^\emptyset_k + \Lambda^\emptyset_k W_k^TW_{\mathcal{C}}
     [P_{\mathcal{C}} +W_{\mathcal{C}}^T(I-W_k\Lambda_k^\emptyset W_k^T)W_{\mathcal{C}}]^{-1} W_{\mathcal{C}}^TW_k \Lambda_k^\emptyset \,,\\
     \Lambda_k^\emptyset & = (W_k^T W_k + P_k)^{-1}\,.
\end{align}
The intermediate quantity $\Lambda_k^\emptyset$ above corresponds to the variance in the update of $q(\theta_k)$ when $\mathcal{C}=\emptyset$, i.e., for the fully factorized variational family, hence the notation. The computation of ELBO also requires:
\[
\log |\Lambda_k| 
=
 \log |\Lambda_k^\emptyset| +
 \log |W_{\mathcal{C}}^TW_{\mathcal{C}} 
 + P_{\mathcal{C}}| - \log|P_{\mathcal{C}} +W_{\mathcal{C}}^T(I-W_k\Lambda_k^\emptyset W_k^T)W_{\mathcal{C}}|. 
\]
The update for $q(\phi)$ is rather standard. In the Gaussian case, we further assume that $q(\Sigma)$ and $q(\sigma^2)$ are independent, implying that $q(\phi) = q(\sigma^2)\prod_{k=1}^K q(\Sigma_k)$. In the binomial case, the model implies that $q(\phi)$ factorizes into $\prod_{i=1}^N q(\omega_i)\prod_{k=1}^K q(\Sigma_k)$ without additional assumptions. 
%First, although not assumed in the factorization, the dependence structure in the model implies that $q(\phi)$ factorizes into independent variational approximations, one for each $q(\Sigma_k)$ and one for $q(\sigma^2)$ in the Gaussian model and one for each $q(\omega_i)$ in the binomial models. 
The coordinate-wise updates are as follows: $q(\Sigma_k)$ has the form of $\mathrm{IW}(a_k, \Phi_k)$, with
\begin{equation*}
		a_k = a^0_k + G_k; \quad \Phi_k = \begin{cases}  \Phi^0_k + \sum_{g=1}^{G_k} \mathbb{E}_{q(\sigma^2)}[1/\sigma^2]~\mathbb{E}_{q(\alpha_{k,g})}\left[\alpha_{k,g} \alpha_{k,g}^T\right] & \textrm{Gaussian model} \\ \Phi^0_k + \sum_{g=1}^{G_k} \mathbb{E}_{q(\alpha_{k,g})}\left[\alpha_{k,g} \alpha_{k,g}^T\right] & \textrm{binomial model}; \end{cases}
	\end{equation*}
while $q(\sigma^2)$ is an $\mathrm{InverseGamma}\left(a_{\sigma^2}, b_{\sigma^2}\right)$ with
	\begin{align*}
	a_{\sigma^2} & = {1\over 2} \left (n + \sum_{k=1}^K D_k G_k\right);\\ b_{\sigma^2} & = {1\over 2} \sum_{i=1}^N \left\{ \mathbb{E}_{q(\theta)}[(y_i -\eta_i)^2] + \sum_{k=1}^K \sum_{g=1}^{G_k} \mathbb{E}_{q(\alpha_{k,g})q(\Sigma_k)}\left[\mathrm{tr}(\Sigma^{-1}_k \alpha_{k,g}\alpha_{k,g}^T)\right]\right\}; 
	\end{align*}
and $q(\omega_i)$ is $\mathrm{PG}(b_i, c_i)$ where $b_i = n_i$ and
$c_i = (E_{q(\theta)}[\eta_i^2])^{1/2}$. The computation of $E_{q(\theta)}[\eta_i^2]$ requires the covariance between random effects in $\mathcal{U}$ and $\mathcal{C}$; this is already obtained when updating $q(\theta_{\mathcal{C}} \mid \theta_{\mathcal{U}})$. The expectations needed to update $\nu$, $D_{ii}^2$ and $T_k$ can be found noting that $\mathbb{E}_{q(\sigma^2)}[1/\sigma^2] = a_{\sigma^2}/b_{\sigma^2}$, $\mathbb{E}_{q(\Sigma_k)}[\Sigma_k^{-1}] = a_k \Phi_k^{-1}$, and $\mathbb{E}_{q(\omega_i)}[\omega_i] = \frac{b_i}{2 c_i} \tanh(c_i/2)$ \citep{polson2013polyagamma}.

\subsection{Computational cost per iteration}\label{sec:comp_cost}

Each iteration of the algorithm requires the Cholesky factor of the  matrix $W_{\mathcal{C}}^T W_{\mathcal{C}} + P_{\mathcal{C}}$, which changes at each iteration and has dimensions $(\sum_{k \in \mathcal{C}} G_kD_k )\times (\sum_{k \in \mathcal{C}} G_kD_k )$.
In the worst case, computing it requires a $\mathcal{O}(\left[\sum_{k \in \mathcal{C}} G_kD_k\right]^3)$ computational cost per iteration. 
%Exploiting the fact that the matrix $W_{\mathcal{C}}^T W_{\mathcal{C}} + P_{\mathcal{C}}$ is usually sparse, 
This cubic cost can be reduced to a linear one for sparse matrices with specific patterns (such as banded matrices, trees, etc), which arise from GLMMs with $K=1$ or purely nested designs with no crossed factors (also called multilevel models, see e.g., \citealt{huang2001multiscale,zhang2008fast,menictas2021streamlined} and references therein).
However, for general designs with crossed effects, the sparsity pattern of the matrix $W_{\mathcal{C}}^T W_{\mathcal{C}} + P_{\mathcal{C}}$ is such that the corresponding Cholesky factor is dense, meaning that sparse Cholesky factorization techniques are not as effective as with, e.g., trees or banded matrices and in general do not reduce the cubic cost (see Section 4 of \citealt{Pandolfi2024} for details, as well as references cited in Section \ref{sec:lit} for additional discussion).

As a consequence, for general designs with potentially crossed factors, the computational cost of coordinate ascent for the unfactorized family, which corresponds to $\mathcal{C} = \{0, \dots, K\}$, scales super-linearly with the number of parameters $p$, see Figure \ref{fig:sim_times} in Section~\ref{sec:numerics} for an illustration. Iterative linear solvers, such as conjugate gradient, could be used to speed up some of the required computations. On the other hand, each iteration of the algorithm also requires log-determinants of $|C| \times |C|$ matrices which are harder to compute. Therefore, we prefer partial factorizations for which  factors with large $G_k$ are not in $\mathcal{C}$.

A naive implementation of coordinate ascent could incur super-linear cost in $\sum_{k \in \mathcal{U}} G_kD_k$. 
In particular, the matrix $\Lambda_k$ is high-dimensional ($G_k D_k \times G_k D_k$) and dense. However, its explicit form can be computed without a high-dimensional matrix inverse by using the Woodbury lemma as $\Lambda_k^{-1}$ is a diagonal matrix plus a $\sum_{k \in \mathcal{C}} G_kD_k$-rank matrix; see (\ref{eq:lambdaWoodbury}); and forming the full $\Lambda_k$ is never explicitly required to run the algorithm. Rather, one needs to compute $\Lambda_k a$ where $a$ is a vector such as that found in (\ref{eq:lambdaMultiply}),  and extract the $G_k$ diagonal blocks of size $D_k$ for use in other steps of the algorithm such as updating $q(\phi)$ and evaluating the ELBO. These blocks can also be extracted by specific (sparse) choices of $a$. Thus, coordinate ascent for a partially factorized family can be implemented without ever forming $\Lambda_k$ explicitly, avoiding potential computational bottlenecks related to those terms.  Our software implementation  (\href{https://github.com/mgoplerud/vglmer}{vglmer}) does precisely this. Therefore, coordinate ascent for a partially factorized family has a cost that scales linearly with $\sum_{k \in \mathcal{U}} G_kD_k$, which is why it can scale well to high-dimensional problems (see Figure \ref{fig:sim_times}).

\section{Scalability theory for mixed models and practical guidance}
\label{sec:theory-GLMM}

\subsection{Convergence and accuracy of  variational inference for mixed models}

Section \ref{sec:theory} shows how the uncertainty quantification fraction defined in \eqref{eq:def_uqf} simultaneously characterizes the approximation error of variational inference and the convergence rate of the associated coordinate ascent algorithm. 
In this section, we study how this behaves for mixed models. Our analysis is based on certain simplifications.  First, as in Section \ref{sec:theory}, we assume that 
the variational approximation of $\phi$, denoted as $q(\phi)$, is fixed and we focus on the sub-routine of the algorithm 
that switches off its update.
By the proof of Proposition \ref{prop:PFCAVI}, this is equivalent to focusing on the intermediate target distribution $\pi(\theta)$ defined in \eqref{eq:tilde_pi_def}. 
Second, we analyze mixed models with $D_k=1$ for $k=0,\dots,K$ (i.e., random-intercept models).
For clarity, we state the resulting model under consideration in a self-contained way below. Recall that $m_{i,k}\in\{0,1\}^{G_{k}}$ is the one-hot-encoded vector that indicates the level of the $k$-th factor to which observation $i$ belongs.

\begin{restrictions}
  \label{ass:glmm-simple}
  Consider the following special case of the model defined in  \eqref{eq:glmm}:  
  \begin{equation}\label{eq:glmm_K}
\begin{aligned}
 y_i \mid \eta_i & \sim p(y_i | \eta_i)\,,\quad  \eta_i  = \theta_0 + \sum_{k=1}^K m_{i,k}^T \theta_k, 
&i=1,\dots,n\\
 \theta_{k,g} &\sim \Nor \left(0, \gamma \Sigma_k\right)
 &g=1,\dots,G_k;\,k=1,\dots,K
\end{aligned}
  \end{equation}
  with $p(y_i|\eta_i)$ and $\gamma$ defined as in either \eqref{eq:gauss} or \eqref{eq:binom}, and a flat prior for $\theta_0$.
Assume $q(\phi)$ is fixed and denote by $\pi(\theta)$ the distribution defined in \eqref{eq:tilde_pi_def}.
Here $\phi=((\Sigma_k)_{k=1}^K, \sigma^2)\in(0,\infty)^{K+1}$ for  Gaussian likelihood in \eqref{eq:gauss}, and $\phi=((\Sigma_k)_{k=1}^K, (\omega_i)_{i=1}^n)\in(0,\infty)^{K+n}$ for binomial likelihood in \eqref{eq:binom}.
\end{restrictions}
We first obtain an upper bound for the $UQF(q^*_{FF} \Vert \pi) $ for fully factorized variational inference, which corresponds to partially factorized with $\mathcal{C} = \emptyset$. 
\begin{theorem}
  \label{th:ff}
Consider Restriction  \ref{ass:glmm-simple} and let $q^*_{FF}$ be the fully factorized variational approximation, i.e., the solution of \eqref{eq:fix-phi} when $\mathcal{C} = \emptyset$. Then
\begin{equation}\label{eq:upper_bound_ffvi}
UQF(q^*_{FF} \Vert \pi) \leq   1 - \max_{1\leq k \leq K} \left( n\bar{D}  \over G_k T_k + n\bar{D} \right)^{1/2}\,,
\end{equation}
where $\bar{D}=n^{-1}\sum_{i=1}^n D^2_{ii}$. Here $(D_{ii})_{i=1}^n$ and $(T_k)_{k=1}^K$ are defined as in \eqref{eq:D_gauss} and \eqref{eq:D_binom}.
\end{theorem} 
Theorem \ref{th:ff} points to convergence and accuracy problems for fully factorized variational inference. 
To see that, consider a high-dimensional regime where both $n$ and $\sum_{k=1}^KG_k$ grow. 
Then, assuming that $\bar{D}$ and $T_k$ are bounded away from $0$ and $\infty$, % --- which is reasonable given their definitions ---, 
the upper bound in \eqref{eq:upper_bound_ffvi} goes to $0$ if and only if $n/G_k\to\infty$ for at least one $k$ and, in such cases, it does at  rate $\mathcal{O}((\min_k G_k)/n)$. 
Thus, in the so-called in-fill large data regime, where the $G_k$'s are fixed and $n$ increases (e.g., constant depth of interactions but increasing number of respondents), $UQF(q^*_{FF} \Vert \pi)$ decreases to $0$ at rate $\mathcal{O}(1/n)$. Additionally, in high-dimensional regimes when the size of the contingency table grows with $n$,  $UQF(q^*_{FF} \Vert \pi)$ usually decreases to 0.
Consider for example the case where $K=2$ and $n$ is a fraction of $G_1 \times G_2$ (say, by observing a completely random subset of the cells). Then if either $G_1$ or $G_2$ increases $UQF(q^*_{FF} \Vert \pi)$ decreases to 0. 
Overall, Theorem \ref{th:ff} suggests that fully factorized variational inference does not typically perform well in large scale mixed models and the numerics of Section \ref{sec:numerics} corroborate this.

With appropriate $ \emptyset \subset \mathcal{C} \subset \{0,1,\ldots,K\}$ we can achieve both accuracy and scalability. We analyse the uncertainty quantification fraction induced by partial factorization for design matrices that satisfy a balancedness condition.
This is formalized in Assumption \ref{ass:balance} in terms of the following weighted counts
  \begin{align}
    n_{g,h}^{(k,\ell)} &= \sum_{i=1}^n D^2_{i,i} m_{i,k,g} m_{i,\ell,h},\quad n_{g}^{(k)}=n_{g,g}^{(k,k)},
    \label{eq:counts}
  \end{align}
  where $k,\ell=1,2$, $g=1,\dots,G_k$ and $h=1,\dots,G_\ell$. 
\begin{assumption}[Balanced designs]  \label{ass:balance}
$K=2$ and 
%  $n^{(k)}_g = \delta/G_k$ 
  $n^{(k)}_g = n^{(k)}_{g'}$ for all $g,g'=1,\ldots,G_k$ and $k=1,\ldots, K$.
\end{assumption}
A balancedness requirement analogous to Assumption \ref{ass:balance} has been previously employed to study Markov chain Monte Carlo for similar models in \cite{papaspiliopoulos2020scalable,papaspiliopoulos2023scalable}.
In the Gaussian case, $n^{(k)}_g$ is equal to a constant,
$\mathbb{E}_{q(\sigma^2)}(1/\sigma^2)$, times the number of observations that involve level $g$ of factor $k$.
Thus, Assumption \ref{ass:balance} requires that each level within a factor receives the same number of observations.
In the binomial case, $n^{(k)}_g$ is equal to the sum of the elements $\mathbb{E}_{q(\omega_i)}[\omega_i]$ over the observations that involve level $g$ of factor $k$. In this case, requiring $n^{(k)}_g$ to be constant across $g$, as in Assumption \ref{ass:balance},  is less realistic. 
Extensions of Theorems \ref{th:pf} and  \ref{th:rg} below to more realistic cases with binomial likelihoods are left to future work.

\begin{theorem}
  \label{th:pf}
  Consider Restriction \ref{ass:glmm-simple} and 
let $q_{PF}^*$ be  %the PF-VI approximation, i.e., 
the solution of \eqref{eq:fix-phi} for $\mathcal{C} = \{0\}$, i.e., when $\theta_{\mathcal{C}}$ contains only the fixed effect $\theta_0$.
%the PF-VI approximation, i.e., $q^*(\theta)$ as in \eqref{eq:fix-phi} with $\mathcal{C} = \{0\}$. 
Then, under Assumption \ref{ass:balance}
\begin{equation}\label{eq:bound_K2}
UQF(q^*_{PF} \Vert \pi) =   1 - \left(\prod_{k=1}^2\left[n\bar{D}  \over G_k T_k + n\bar{D}  \right]^{1/2}\right)\lambda_{aux}^{1/2}\,, 
\end{equation}
with $\lambda_{aux}$ being the second largest modulus eigenvalue of the $G_1\times G_1$ stochastic matrix $S_{12}S_{21}$, where $S_{k\ell}$ for $k\neq \ell$ is a $G_k\times G_\ell$ stochastic matrix with $(g,h)$-th entry  equal to  $ n_{g,h}^{(k,\ell)}/n_{g}^{(k)}$.
\end{theorem}

Given that $\lambda_{aux}$ is a modulus eigenvalue of a stochastic matrix, we have $\lambda_{aux}\in[0,1]$. Thus \eqref{eq:bound_K2} implies
\begin{equation}\label{eq:loose_bound}
UQF(q^*_{PF} \Vert \pi) \geq  1 - \left(\prod_{k=1}^2\left[n\bar{D} \over G_k T_k + n\bar{D} \right]^{1/2}\right).
\end{equation}
Combined with \eqref{eq:upper_bound_ffvi}, the inequality in \eqref{eq:loose_bound} implies that $UQF(q^*_{PF,\{0\}} \Vert \pi)$ is always larger than $UQF(q^*_{FF} \Vert \pi)$, meaning that partial factorization with $\mathcal{C} = \{0\}$ always performs better than $\mathcal{C} = \emptyset$ in this context. 
However, the lower bound in \eqref{eq:loose_bound}, which ignores the term $\lambda_{aux}$, is usually far from tight, since often $\lambda_{aux}\ll 1$.
In particular, assuming $\bar{D}$ and $T_k$ bounded away from $0$ and $\infty$ as before, the lower bound in \eqref{eq:loose_bound} goes to $0$ when $n/G_k\to\infty$ for both $k=1,2$.
One thus needs to consider the term $\lambda_{aux}$ more carefully to gain insight into the behaviour of partially factorized variational inference.
The value of $\lambda_{aux}$ depends on the  matrix $N_{1,2}=(n_{g,h}^{(1,2)})_{g,h}$, which we refer to as the design or co-occurrence
%\emph{design} or \emph{co-occurrence} 
matrix between factor $1$ and $2$, 
and it can be interpreted as a measure of graph connectivity. 
Specifically, define $\mathcal{G}$ as a weighted graph with $G_1+G_2$ vertices, which are in one-to-one correspondence with the parameters in $(\theta_1,\theta_2)$ such that $\mathcal{G}$ is bipartite (i.e., there is no edge within elements of $\theta_{1}$ nor within elements of $\theta_{2}$) and the adjacency matrix between $\theta_1$ and $\theta_2 $ is $N_{1,2}$. In other words, there is an edge between ${\theta}_{1,g}$ and ${\theta}_{2,h}$ if and only if $n^{(1,2)}_{g,h}\neq 0$ and $n^{(1,2)}_{g,h}$ is the weight of such edge.
Then, $1-\lambda_{aux}^{1/2}$ coincides with the spectral gap of the (normalized) Laplacian of $\mathcal{G}$. 
Thus, the smaller $\lambda_{aux}$, the better the connectivity of $\mathcal{G}$.
By \eqref{eq:bound_K2} we have
\begin{equation}\label{eq:bound_lambda_aux}
UQF(q^*_{PF} \Vert \pi)\geq 1 - \lambda_{aux}^{1/2}\,, 
\end{equation}
meaning that good connectivity of $\mathcal{G}$ is sufficient for partially factorized variational inference with $\mathcal{C}=\{0\}$ to perform well.

Given the connection between $UQF(q^*_{PF} \Vert \pi)$ and $\lambda_{aux}$ in Theorem \ref{th:pf}, we can leverage known results from spectral random graph theory \citep{brito2022spectral} to provide a sharp analysis of $UQF(q^*_{PF} \Vert \pi)$ under the assumption of a random design.
Roughly speaking, we will assume that $N_{1,2}$ is sampled uniformly at random from the designs that satisfy Assumption \ref{ass:balance}, 
which corresponds to assuming that $\mathcal{G}$ coincides (after normalization) with a random bipartite biregular graph. 
More precisely, let $d_{1}$, $d_{2}$ and $n$ be non-negative integers such that $n$ is a multiple of both $d_{1}$ and $d_{2}$, and let $G_{1}= n/d_{1}$, $G_{2} = n/d_{2}$.
Define $\mathcal{D}(n, d_{1}, d_{2})$ as the collection of all $G_{1}\times G_{2}$ binary 
designs with $n$ observations satisfying the balancedness condition of Assumption \ref{ass:balance}, i.e.,
$$
\mathcal{D}(n, d_{1}, d_{2})=\{N_{1,2}\in\{0,\bar{D}\}^{G_1\times G_2}\,:\,n^{(k)}_g = \bar{D}d_k\hbox{ for all }k=1,2\hbox{ and }g=1,\ldots,G_k\}\,.
$$
Assumption \ref{ass:random_design} formalizes our random design assumption.
\begin{assumption}[Random designs]\label{ass:random_design}
$N_{1,2}\sim \mathrm{Unif} (\mathcal{D}(n, d_{1}, d_{2}))$.
 \end{assumption}

\begin{theorem}
  \label{th:rg}
  Consider Restriction \ref{ass:glmm-simple} and let $q^*_{PF}$ be the solution of \eqref{eq:fix-phi} for $\mathcal{C} = \{0\}$.
Then, under Assumption \ref{ass:random_design}, for every $\epsilon>0$ it holds that
\begin{equation}\label{eq:lower_random}
UQF(q^*_{PF} \Vert \pi)
\geq
1 - ((G_1/n)^{1/2}+(G_2/n)^{1/2}))^{1/2}-\epsilon
\end{equation}
almost surely as $n\to\infty$ (with fixed $d_1=n/G_1$ and $d_2=n/G_2$).
\end{theorem}
Theorem \ref{th:rg} implies that, as $n,G_1,G_2\to\infty$ with fixed ratios $n/G_1$ and $n/G_2$, $UQF(q^*_{PF,\{0\}} \Vert \pi)$ remains lower bounded by a strictly positive constant.
In other words, taking $\mathcal{C} = \{0\}$ is enough 
to ensure that the uncertainty quantification of partially factorized variational inference does not deteriorate in such high-dimensional mixed models. 
Also, note that the lower bound in \eqref{eq:lower_random} gets arbitrarily close to $1$ as $n/G_1$ and $n/G_2$ increase.
This suggests that, if $n,G_1,G_2\to\infty$ with $n$ growing faster than $G_1$ and $G_2$, $UQF(q^*_{PF} \Vert \pi)$ goes to $1$. 
We refer to this phenomenon as blessing of dimensionality, since the algorithm performs increasingly better as the amount of data and dimensionality of the problem increases \citep{brito2022spectral}. 
From the mathematical point of view, this phenomenon is a consequence of random graphs being optimally connected (e.g., being expanders) with high probability as $n$ increases. The numerics of Section \ref{sec:numerics} corroborate this.

\subsection{Practical recommendations}\label{sec:recomm}
Theorem \ref{th:rg} shows that, under the assumption of a random design, the graph $\mathcal{G}$ is well-connected with high probability and, as a result, $UQF(q^*_{PF} \Vert \pi)$ is far from $0$. In those situations, it therefore suffices to include in $\mathcal{C}$ only the fixed effects. 
On the other hand, for data designs that are not well-described by random graphs, $\mathcal{G}$ can be poorly connected and $\lambda_{aux}$ be close to $1$. 
A practically important example is the case when one factor is nested into another. One example arises when $K=2$, $G_2=rG_1$ for some $r\in\mathbb{N}$ and the design matrix $N_{1,2}$ is such that $n_{g,h}^{(1,2)}=1$ if $r(g-1)<h\leq rg$ and $0$ otherwise. 
Nested designs of this type commonly arise when one factor is defined as the interaction between other factors; Section~\ref{sec:ghitza} considers this example in detail. In such settings, it typically holds by construction that $\mathcal{G}$ is disconnected and thus $\lambda_{aux}=1$. This means that the inequality in \eqref{eq:loose_bound} becomes an equality and, as a result, $UQF(q^*_{PF} \Vert \pi)$ goes to $0$ whenever $n/G_k\to\infty$ for both $k=1,2$. In that setting taking $\mathcal{C}=\{0\}$ is not sufficient for partially factorized variational inference to achieve good accuracy and the corresponding coordinate ascent algorithm to be scalable. A more statistical intuition is that  nesting creates a strong posterior dependence between $\theta_1$ and $\theta_2 $, 
even after marginalizing out ${\theta}_0$. 

In such cases one needs to increase the size of $\mathcal{C}$. 
For example, in the nested design discussed above where $K=2$ and $I_2=rI_1$, setting $\mathcal{C}=\{0,1\}$ results in $UQF(q^*_{PF} \Vert \pi)=1$, which can be easily deduced from the proof of Proposition \ref{prop:uqf-bound-pf} and the fact that $|\mathcal{U}|=1$ in this case. In more general settings,  increasing the size of $\mathcal{C}$ is potentially undesirable as it must increase the computational burden per iteration due to the  $\mathcal{O}(\left[\sum_{k \in \mathcal{C}} G_kD_k\right]^3)$ costs discussed in Section \ref{sec:comp_cost}.

Interestingly, the connection between the uncertainty quantification fraction and the connectivity of the graph $\mathcal{G}$ suggests a simple strategy to obtain  accurate and scalable variational schemes for models with interaction terms: Include the main effects, i.e., those on which interactions are constructed from, alongside the fixed effects $\theta_0$ in $\mathcal{C}$. More generally, in the case of nested designs, we recommend including in $\mathcal{C}$  an effect $\theta_k$ if (and only if) there are other effects $\theta_{k'}$ which are nested inside of $\theta_k$. 
This strategy ensures that the mean-field assumptions is employed only across factors that are not nested within one another (and thus where a random graph model is more plausible description for the co-occurrence matrix). 
Note that, by design, main effects tend to have a much smaller size than interaction ones (since the size of the latter is a product of sizes of the former) and thus the computational cost  with such strategy will typically be much smaller than setting $\mathcal{C}=\{0,\dots,K\}$. Section \ref{sec:ghitza} explores this strategy and demonstrates that it is preferable to other specifications, in terms of an accuracy-computational cost trade-off.

Finally, it is important to note that our theoretical results about $UQF(q^*_{PF} \Vert \pi)$ 
apply only to $K=2$ and $\mathcal{C}=\{0\}$ 
and thus do not directly cover more complex cases discussed in this section (e.g., those with $K>2$, as it happens in the presence of interactions terms). On the other hand, the above theory gives insight into the type of dependencies that can reduce the effectiveness and allows to design strategies to overcome those. 
Crucially, the numerics of Section \ref{sec:numerics} suggest that the main conclusions of Theorems \ref{th:ff}, \ref{th:pf} and \ref{th:rg}  extrapolate well to more complex situations, such as those with $K>2$, unbalanced designs and random-slopes, and provide useful predictions about the performance of partially factorized variational inference  in real-data examples.

\section{Connections to  alternative variational inference frameworks}\label{sec:lit}

Various works have proposed relaxations of the mean-field assumption,  where some conditional dependencies are retained and others ignored,  in order to improve the accuracy of variational inference,  see, e.g., the work on structured variational inference \citep{saul1995exploiting,salimans2013fixed,hoffman2015structured} or collapsed variational inference \citep{teh2006collapsed}.
These methods have been widely and successfully applied to nested hierarchical models, i.e., models with observations $y=(y_1,\dots,y_n)$, local latent variables $z=(z_1,\dots,z_n)$, global parameters $\xi$, and joint distribution factorizing as $p(\xi,z,y)=p(\xi)\prod_{i=1}^n p(y_i, z_i\mid \xi)$, see, e.g., \citet{hoffman2015structured,loaiza2022fast,fasano2022scalable}.
This framework includes mixed models with one set of nested effects, such as model \eqref{eq:glmm} with $K=1$, for which multiple variational schemes have been proposed \citep{ncp-vi,tan2018gaussian,tan2021rvb}.
Relaxations of mean-field have also been successfully applied to time-series models 
\citep{saul1995exploiting,salimans2013fixed,tan2018gaussian,frazier2023variational}. In this paper, instead, we consider the case of general mixed models, with multiple sets of random effects (i.e., $K>1$ in model \eqref{eq:glmm}) including both nested and crossed effects. %\citep{baayen2008mixed}. 
The resulting posterior dependence structure is computationally more challenging than the one resulting from nested or Markov models, since one cannot directly take advantage of efficient algorithms and computational approaches designed for, e.g., trees or banded matrices, see the discussions in \cite{GaoOwen2017EJS,menictas2023streamlined,papaspiliopoulos2023scalable} as well as Section \ref{sec:comp_cost}.

Contributions analyzing mixed models with non-nested factors include \cite{menictas2023streamlined}, who considers variational inference with blocking, i.e., grouping together some parameters in the mean-field factorization, for two-factor crossed effect models with Gaussian likelihood, and \cite{goplerud2022mavb}, who considers blocking and Polya-Gamma data augmentation for  binomial likelihood.  Using our notation, the blocking analogue of $\mathcal{Q}_{PF}$ would be the family $\mathcal{Q}_{B}= \left\{q\,:\
q(\theta,\phi) = q\left(\theta_{\mathcal{C}}\right) \left[\prod_{k \in \mathcal{U}} q(\theta_{k})\right] q(\phi)\right\}$.
While helpful in improving accuracy, blocking does not manage to simultaneously achieve non-decreasing accuracy in high-dimensions and linear computational cost in $n$ and $p$, see e.g., the comparison and tension between factorizations I, II and III discussed in  \cite{menictas2023streamlined}.
Thus, blocking does not provide the same level of computational-accuracy trade-off as our proposed partial factorization. 
Also, none of these works provide theoretical guarantees or analysis in high-dimensional regimes.
Recently, \cite{xu2023gaussian} considers using composite likelihood approximations to reduce computational cost of variational inference for two-factor crossed effect models with Gamma or Poisson likelihoods and complete observations.

Previously proposed structured variational inference methods allow high-dimensional local parameters to depend on low-dimensional global ones \citep{hoffman2015structured,tan2018gaussian,loaiza2022fast}. 
While this is a computationally feasible and effective strategy for nested models, to the best of our knowledge, it is not clear how to make it work successfully for general mixed models with $K>1$.
Instead, we propose using the reverse factorization, where a set of low-dimensional (typically global)  parameters can depend on the remaining high-dimensional ones. We show theoretically and numerically that this is both computationally feasible, in the sense of implementable with total cost scaling linearly with $n$ and $p$ (Section \ref{sec:comp_cost}), and leads to drastically improved accuracy in high-dimensions (Sections \ref{sec:theory-GLMM} and \ref{sec:numerics}).

In this paper, we employ Polya-Gamma data augmentation \citep{polson2013polyagamma} for binomial mixed models to restore conditional conjugacy and allow for closed-form coordinate ascent updates. The augmentation also facilitates the theoretical analysis in Section  \ref{sec:theory-GLMM}.  
This data augmentation strategy is equivalent to using the popular tangent bound by \cite{jaakkola2000bayesian} for variational inference for logistic regression \citep{durante2019conditionally}.
Another popular alternative is to take a so-called fixed-form variational approach \citep{salimans2013fixed,titsias2014doubly,kucukelbir2017automatic}, seeking to minimize the divergence between the original posterior (without additional latent variables) and a prescribed parametric family of distributions (e.g., Gaussians with diagonal or low-rank covariance). % and , see e.g., Gaussian VI approach of \citep{ormerod2012gva} applied to nested models. 
Given the lack of conditional conjugacy, this approach either requires (low-dimensional) numerical integration or stochastic optimization, see, e.g., \citet{ormerod2012gva,jeon2017variational,tan2018gaussian,castiglione2022bayesian} for applications to mixed models.
In the Supplement, we numerically compare our proposed approach with two popular fixed-form variational methodologies implemented in STAN, namely 
ADVI \citep{kucukelbir2017automatic} and Pathfinder \citep{zhang2022pathfinder}. Section \ref{sec:summary_add_numerics} provides a summary of the results.

\section{Numerical experiments}
\label{sec:numerics}

\subsection{Empirical estimation of accuracy and stopping rules}
\label{sec:acc}

We  estimate $UQF(q^*(\theta) \Vert p(\theta|y))$, i.e., the uncertainty quantification fraction on the marginal with respect to $\phi$  distribution on fixed and random effects under the posterior distribution and its variational approximation, exploiting its characterisation as $1/\lambda_{max}(\mathrm{cov}_\pi(\theta)\mathrm{cov}_{q^*}(\theta)^{-1})$ established in Proposition \ref{prop:uqf}. The matrix $\mathrm{cov}_{q^*}(\theta)^{-1}$ is directly available using quantities already computed during the coordinate ascent iterations. The  matrix $\mathrm{cov}_\pi(\theta)$ is estimated using Markov chain Monte Carlo samples from the true posterior, drawn using the Gibbs samplers proposed in \cite{papaspiliopoulos2023scalable}, with Polya-Gamma augmentation for a binary outcome, for the experiments in Section \ref{sec:simstud} and Hamiltonian Monte Carlo for those in Section \ref{sec:ghitza} using the code from \cite{goplerud2022mavb}. % in order to  draw samples from the true posterior and use them as a benchmark for the variational methods. 

However, since we are dealing with high-dimensional distributions, naively replacing $\mathrm{cov}_\pi(\theta)$ with its sample estimate would incur a large bias when estimating eigenvalues \citep[ch. 6]{wain}. Instead, we follow a split-sample approach where we split the samples into five folds; using four folds, we form a sample estimate $\widehat{\mathrm{cov}_\pi(\theta)}$ and compute the top 50 eigenvectors of  $\widehat{\mathrm{cov}_\pi(\theta)}\mathrm{cov}_{q}(\theta)^{-1}$ where $V$ denotes the matrix of those eigenvectors. Then, using the held-out 20\% of the samples, we compute a sample estimate of $\mathrm{cov}_\pi(V^T\theta)$, which is low-dimensional, and use it to estimate $\lambda_{max}(\mathrm{cov}_\pi(V^T\theta)\mathrm{cov}_{q}(V^T\theta)^{-1})^{-1}$;  we repeat this for each fold and average the estimates. 
%We numerically assess the bias and variance of the resulting estimates and chose the number of MCMC samples large enough
%The resulting bias and variance estimators exhibit much smaller

We also consider an alternative, more direct, measure of accuracy for scalar-valued functions of the parameters. For a  parameter of interest $x$, we estimate a  linear transformation of the total variation distance of its distribution under the target and its variational approximation:  $1-\frac{1}{2} \int_{-\infty}^\infty | q(x) - \pi(x)|\mathrm{d} x$ (\citealt{faes2011variational}). We compute this as follows: using samples from $\pi$ and $q$, we build a binned kernel density approximation of the densities \citep{wand1994kernel} from which the total variation distance can be easily computed.

Finally, for the coordinate ascent algorithms, we stop  when the change in successive iterations of the ELBO is below a small absolute threshold ($10^{-6}$).

\subsection{Simulation studies}
\label{sec:simstud}

We first evaluate our methodology using the following simulation study. We consider the intercept-only model in \eqref{eq:glmm_K}, both with Gaussian and binomial likelihood, for $K=2$, $G_1=G_2$, and unknown $\phi $.
The co-occurrence matrix  $N_{1,2}$ is taken to be binary with entries
missing completely at random with probability 0.9, and when there is an observation it is generated according to the model (with $\sigma=1$ in the Gaussian model, and $n_i=1$ in the binomial model). We set the prior on $\Sigma_k$ to be $\mathrm{InverseGamma}(1, 0.5)$.

We consider values of $G_1=G_2$ in $(2^5,2^6,\dots,2^{10})$. Note that $n$ also increases by design with increasing $G_1=G_2$. 
Figure \ref{fig:uqf_sim} reports the estimated uncertainty quantification fraction for the three variational approximating families, where results are averaged over 100 simulated datasets, while Figure \ref{fig:sim_times} reports the average time per iteration for running coordinate ascent until the stopping criterion. Number of parameters refers to the total number of random effects and fixed effects, i.e., $1+G_1+G_2$. 

\begin{figure}[!htbp]
  \caption{UQF for simulation study}\label{fig:uqf_sim}
  \begin{center}
    \includegraphics[width=0.8\textwidth]{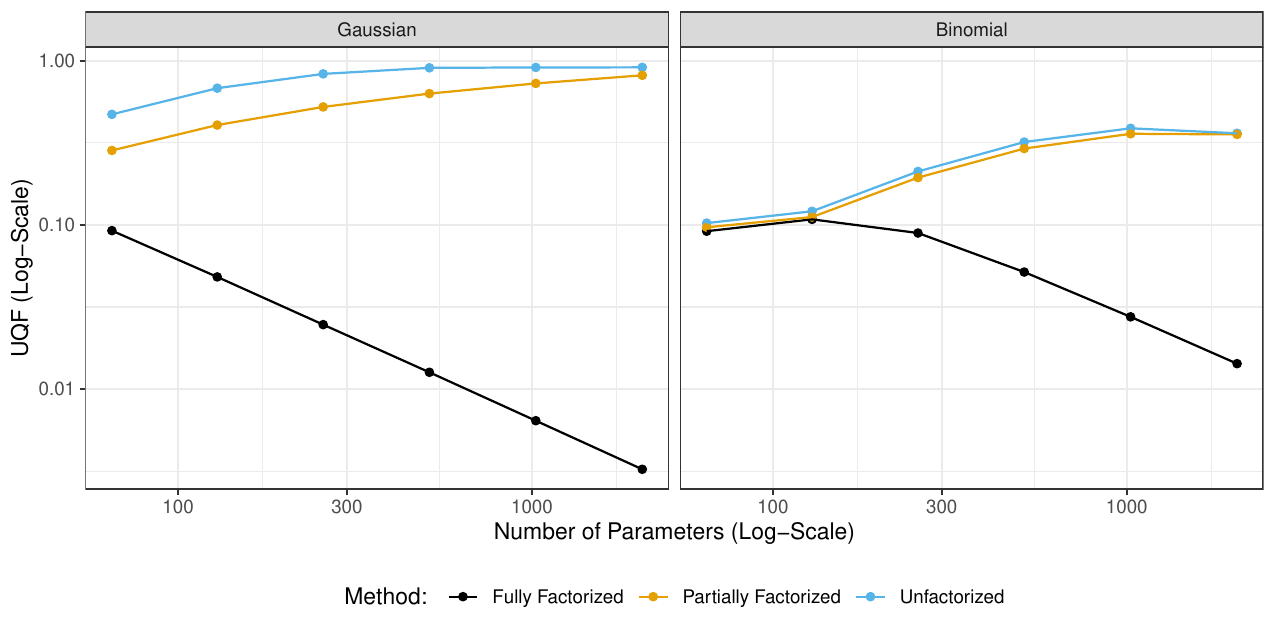}
  \end{center}
\end{figure}
\begin{figure}[!htbp]
  \caption{Time per iteration for simulation study}\label{fig:sim_times}
  \begin{center}
    \includegraphics[width=0.8\textwidth]{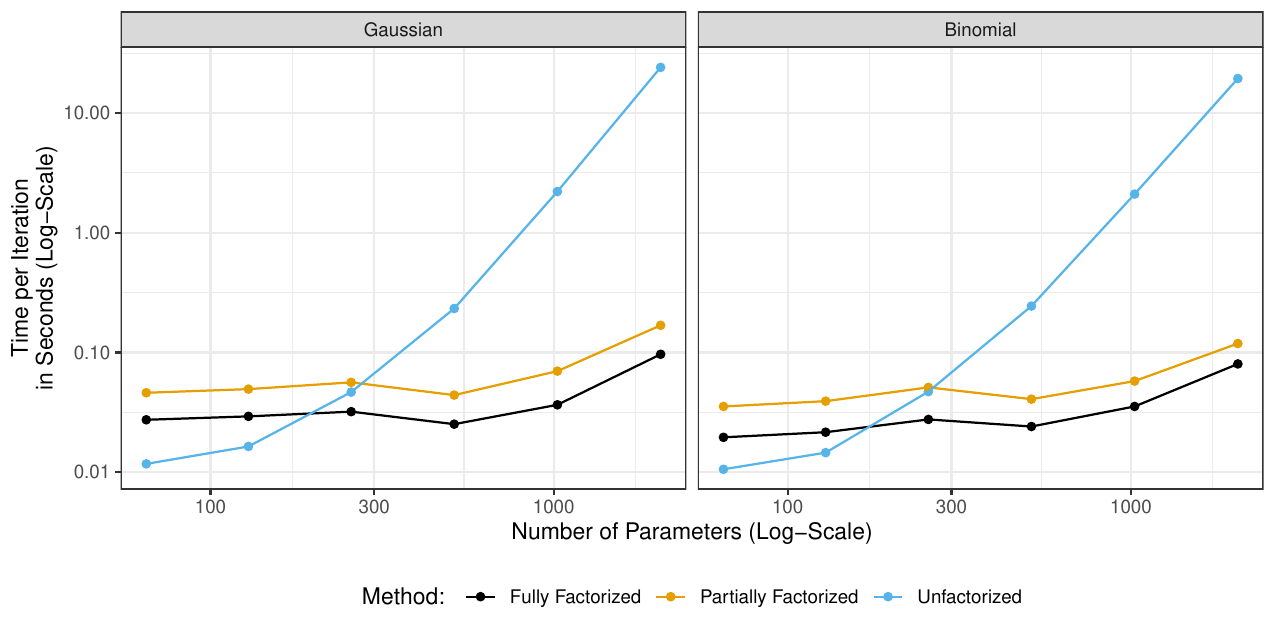}
  \end{center}
\end{figure}

The simulation design violates various assumptions made in Section \ref{sec:theory-GLMM}. In particular, $q(\phi)$ is updated and not kept fixed, and the design matrices are not balanced in the sense of Assumption \ref{ass:balance}. Nevertheless, the numerical results of Figure \ref{fig:uqf_sim}  are in very good qualitative accord with the theorems, e.g. the rate of decay of $UQF(q_{FF}^*(\theta) \Vert p(\theta\mid  y))$ and the increase of $UQF(q_{PF}^*(\theta) \Vert p(\theta\mid y))$ with dimensionality. With the binomial likelihood, performance is worse, possibly due to the approximations needed to enable closed-form coordinate ascent updates; partially factorized variational inference, however, still does enjoy the blessing of dimensionality.

Figure~\ref{fig:sim_times} shows that for very small models, unfactorized variational coordinate ascent is faster per iteration than the other implementations; this is likely due to the specific implementation of the algorithms, e.g., for loops required for the latter and are avoided in the former.  However, for larger models, the computational cost per iteration of unfactorized grows super-linearly, as expected.

\subsection{Modelling voter turnout using multilevel regression and post-stratification}
\label{sec:ghitza}

Multilevel Regression and Post-Stratification \citep{park2004mrp} is  a  method for small-area estimation using survey data that is increasingly popular in political science (e.g., \citealt{lax2009estimation,warshaw2012district,broockman2018bias,lax2019party}) and  proceeds as follows. The combination of the categorical covariates in the survey 
define types of respondents (hereafter referred to as ``types''), e.g.,  white male voters, aged 18-29, with high-school education in Pennsylvania.
A model as in \eqref{eq:glmm} is trained on the survey data and predicts the outcome given a type of respondent. 
Then, census data is typically used to derive the number of people of a given type in each geographic location, although see, e.g., \cite{kastellec} for other possibilities. The geographic-level estimate is then obtained by taking a weighted average of the model predictions for each type using the census frequencies. 
\cite{ghitza} propose using deep models that consist of many interactions between the different demographic and geographic categorical factors recorded in the survey; \cite{goplerud2023reeval} discusses the advantages of these models compared to other machine learning approaches for building predictive models using survey data. 
\cite{goplerud2022mavb}  develops a framework for variational inference that implements either fully factorized or unfactorized variational inference, using the terminology we establish in this article,  and applies it to the data in \cite{ghitza}. There, it is observed that fully factorized variational inference  has poor uncertainty quantification while computational cost of unfactorized variational inference becomes too large for the deeper models.

This section uses the models and design in \cite{goplerud2022mavb}. The response variable is binary and records whether a respondent voted in previous elections. It considers nine model specifications with  increasingly complex interactions among basic categorical covariates: the first model is a simple additive model with four categorical factors (age, ethnicity, income group, and state) and the ninth model is highly complex with eighteen categorical factors (including, e.g., all two-way and three-way interactions among the  original four factors). All models also include some fixed effects, see  Table~\ref{tab:gg_models_main}  summarizes the number of random effects involved in these models, and more details are provided in the Supplement.  \cite{ghitza} consider both the 2004 and 2008 elections separately; we follow this but pool or average the results from each election in the following figures for simplicity. In each model, there are 4,080 types; in each election, and a dataset of around 75,000 (weighted) observations.

\begin{table}
  %\spacingset{1}
  \tbl{Brief description of models for voter turnout}
{  \begin{tabular}{lrrrrrrrrr}
    Model & 1 & 2 & 3 & 4 & 5 & 6 & 7 & 8 & 9\\
    $dim(\theta_{\mathcal{C}})$ & 6 & 6 & 27 & 129 & 139 & 139 & 139 & 139 & 139 \\
    $dim(\theta_{\mathcal{U}})$ & 64 & 123 & 158 & 719 & 784 & 864 & 1884 & 2700 & 3720
  \end{tabular}
}
	\begin{tabnote}
Increasing indices correspond to deeper interactions; models 1 and 2 include  basic categorical factors, from model 3 two-way interactions of categorical factors are introduced, from model 6 three-way interactions; $dim(\theta_{\mathcal{C}})$ is the number of parameters partially factorized and includes fixed and main effects; the unfactorized algorithm decomposes matrices of size $dim(\theta_{\mathcal{C}})+dim(\theta_{\mathcal{U}})$, whereas the partially factorized of size $dim(\theta_{\mathcal{C}})$. 
	\end{tabnote}
         \label{tab:gg_models_main}
\end{table}

For the partially factorized approximations, we follow the principle suggested in Section \ref{sec:recomm} and include fixed and, when interactions exist, the corresponding main effects in $\mathcal{C}$. We compare the variational approximations  against a gold standard posterior sampled using Hamiltonian Monte Carlo using the NUTS/STAN engine. 
Paralleling the analyses in Section~\ref{sec:simstud}, we present two sets of results to illustrate the relative performance of the methods.  Figure~\ref{fig:uqf_gg} shows, as predicted by the theory, that the uncertainty quantification fraction for fully factorized is close to zero for all models whereas that for partially factorized remains drastically larger even for the most complex models. We also see that adding  the main effects to $\mathcal{C}$ makes a significant difference  when there are interactions in the model (i.e., for models 3 to 9). 
\begin{figure}[!htbp]
  \caption{$UQF$ for voter turnout}\label{fig:uqf_gg}
  \begin{center}
    \includegraphics[width=0.8\textwidth]{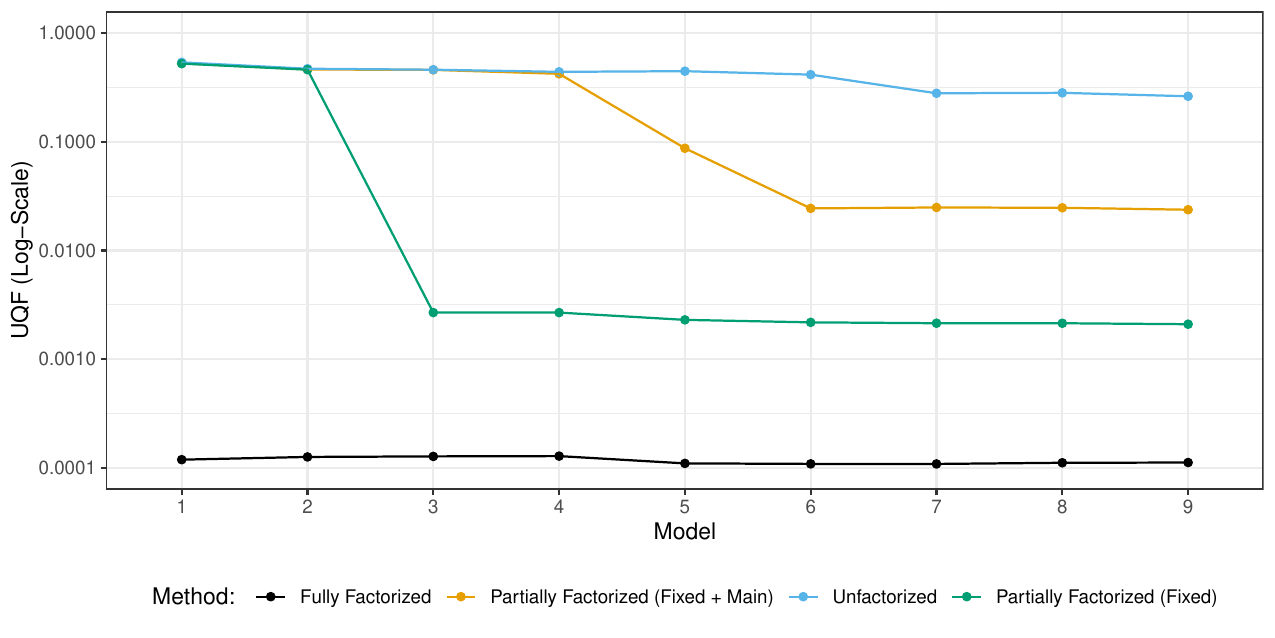}
  \end{center}
\end{figure}
We also consider the accuracy of the variational approximations on key quantities of interest to applied researchers.  One is  state-level estimates of turnout that are computed using post-stratification.  Another is  the linear predictor for each of the  4,080 types.
Figure~\ref{fig:acc_gg} provides  boxplots of estimated  accuracy scores across all parameters of interest and datasets. We see that, for the state-level estimates, partially factorized variational inference with main effects are included in $\mathcal{C}$ performs well even for the most complex models. For the linear predictor, the results are roughly analogous, with the only exception being that partially factorized variational inference with main effects included in $\mathcal{C}$ is accurate for most types for models from 1 to 6, but a tail of inaccurate types appears for models 7 to 9. 
Investigating this phenomenon suggested that this tail typically corresponds to rare types, i.e. cells in the contingency table with very few observations. As these observations are downweighted during the post-stratification stage, it is still possible to see the high levels of accuracy found in the left panel. On the other hand, we have empirically found that full factorization, in addition to underestimating the posterior variance for rare types, tends to dramatically overestimate the variance of common types, which in turns contributes to the low accuracy on the left panel.  

\begin{figure}[!htbp]
  \caption{Accuracy of estimates for voter turnout}\label{fig:acc_gg}
  \begin{center}
    \includegraphics[width=0.8\textwidth]{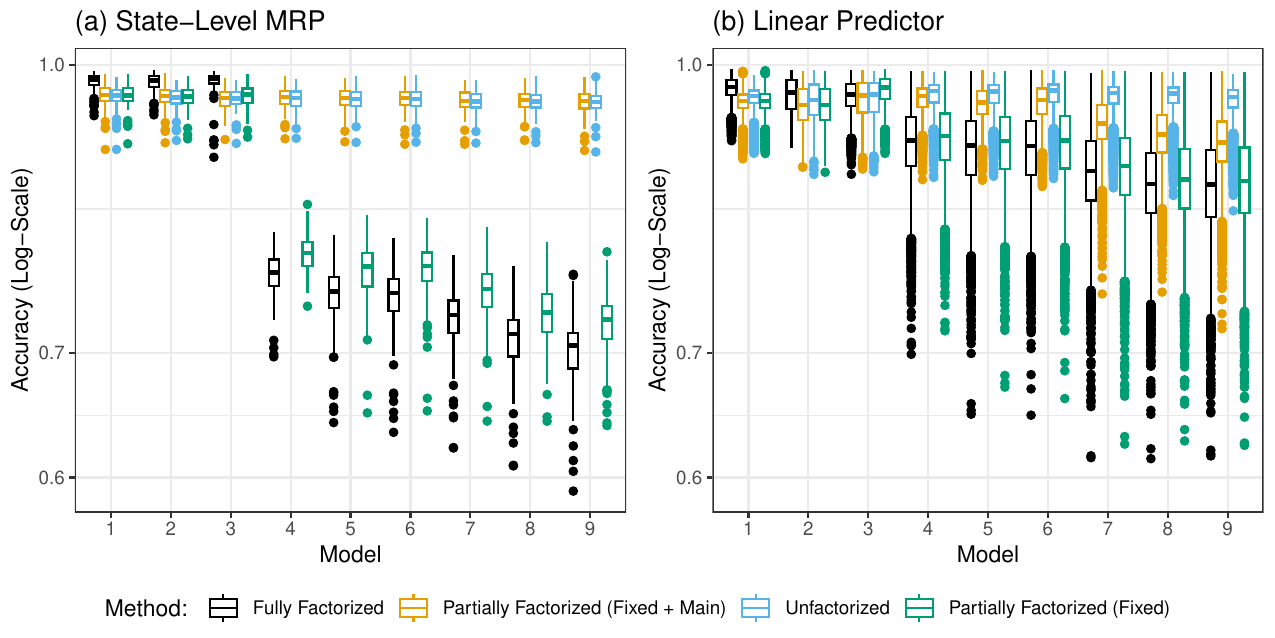}
  \end{center}
\end{figure}

Overall,  for the majority of models and quantities of interests, partially factorized variational inference with $\mathcal{C}$ chosen as proposed in Section \ref{sec:recomm} achieves an accuracy comparable to unfactorized variational inference but at a fraction of computational time. Using the stopping criterion noted above, estimation using partially factorized variational inference (including main effects in $\mathcal{C}$) takes around 1.5 minutes for the most complex model (Model 9), versus around 20 minutes for unfactorized variational inference.

\subsection{Additional numerical comparisons}\label{sec:summary_add_numerics}

This section summarizes two additional sets of numerical comparisons that are reported in detail in the Supplement. We first compare against fixed form variational inference and the implementations in ADVI \citep{kucukelbir2017automatic} and Pathfinder \citep{zhang2022pathfinder}. For Gaussian likelihood, we find that ADVI, which relies on Gaussian distributions whose covariance is restricted to be diagonal, performs similarly to fully factorized variational inference, whereas  
Pathfinder, which relies on Gaussian distributions with diagonal plus low-rank covariance matrices, displays a more constant uncertainty quantification fraction as dimensionality increases, albeit considerably lower than the one estimated from partially factorized variational inference. 
For binomial likelihood,  ADVI and Pathfinder perform better  than the methods we consider in this article in smaller-scale problems, indicating that there is some penalty to using the data-augmentation approach we adopt. However, as dimensionality increases they both give worse results than partially factorized variational inference.
On the real data example partially factorized variational inference consistently out-performs both ADVI and Pathfinder across all metrics that we consider in this paper. Pathfinder's performance appears most similar to partially factorized variational inference where only the fixed effects are included in $\mathcal{C}$. We also investigate thoroughly accuracy and convergence for the variance components, the terms that are included in the vector $\phi$ using our notation. In the simulations, there are modest differences between full, partially, and unfactorized families for less complex models; for larger models, the accuracy of all three increases and the differences between them decrease, although fully factorized variational inference remains inferior to partially factorized which is itself inferior to unfactorized. In the Ghitza-Gelman example, fully factorized performs poorly, while unfactorized is very similar to partially factorized with the fixed and main effects in  $\mathcal{C}$.

\section{Ongoing and future work}
\label{sec:conc}

It is possible to apply our theoretical framework to study related variational schemes discussed in Section \ref{sec:lit}, for example blocking. It is interesting to try and extend it to more general settings, such as log-concave distributions. Although  we have focused on the use of variational fully Bayesian inference, our methods and theory can appropriately be modified for   different inferential frameworks, e.g., for restricted maximum likelihood using a variational expectation-maximization algorithm or methods based on Laplace approximations. We also expect our methods to extend to  other likelihood functions using, for example,  non-conjugate variational message passing \citep{knowles2011non,wand2014fully,ncp-vi} or Gaussian variational approximations \citep{ormerod2012gva}. On the modelling side, our framework immediately extends to more flexible, and often inferentially more appropriate, priors on $\Sigma_k$  \citep{goplerud2023reeval}. The accompanying R package \href{https://github.com/mgoplerud/vglmer}{vglmer} includes a  more flexible partial factorization scheme where  $\mathcal{C}$  includes only some levels $g$ of a factor $k$ in $\mathcal{C}$. One could also use data-driven criteria to decide automatically on $\mathcal{C}$. Finally,  additional acceleration techniques, e.g., parameter-expansion   \citep{goplerud2023reeval} can be incorporated. 
 
\section*{Supplementary Materials}

Supplementary material available at \emph{Biometrika} online includes proofs of propositions and theorems, details of the random scan coordinate ascent algorithm, and additional figures and details for the numerical experiments in Section~\ref{sec:numerics}.

\bibliographystyle{biometrika}
\bibliography{biblio}

\clearpage
\setcounter{section}{0}
\setcounter{figure}{0}
\setcounter{table}{0}
\setcounter{theorem}{0}
\setcounter{equation}{0}
\setcounter{proposition}{0}
\setcounter{assumption}{0}
\setcounter{definition}{0}
\setcounter{proof}{0}
\setcounter{restrictions}{0}

\begin{abstract}
	This supplementary material contains additional proofs, figures, and numerical results for the paper ``Partially factorized variational inference for high-dimensional mixed models.''
\end{abstract}

\section{Proofs of main results}

\subsection{Proofs of Section \ref{sec:theory}}

\begin{proof}[of Proposition \ref{prop:uqf}]
	To simplify exposition let $\Sigma,\widetilde{\Sigma}$ be the covariance matrices of $\theta$ under $\pi$ and $q$ respectively and $Q,\widetilde{Q}$ the corresponding precision matrices. 
	Given that $\widetilde{\Sigma}$ is positive definite a square root $\widetilde{\Sigma}^{1/2}$ is well-defined and without loss of generality we will take below the symmetric square root. 
	Notice first that $\mathrm{var}_{q}(\theta^T v) = v^T \tilde{\Sigma} v$, with an analogous expression for $\mathrm{var}_\pi(\theta^T v)$. Then, 
	since the supremum over all $v \in \mathbb{R}^d$ is the same over all vectors $\widetilde{\Sigma}^{1/2} v$ (due to the assumed invertibility of the matrix) we have that
	$$
	UQF(q \Vert \pi) = \inf_{u} {u^T u \over u^T \widetilde{Q}^{1/2} \Sigma\widetilde{Q}^{1/2} u}= 1 / \sup_{u: ||u||=1} u^T \widetilde{Q}^{1/2} \Sigma \widetilde{Q}^{1/2} u.
	$$
	Recall now that for a positive definite matrix the supremum of quadratic forms over all vectors of norm 1 is its maximum eigenvalue, and that the maximum eigenvalue of a positive definite matrix is the inverse of the minimum eigenvalue of its inverse. Recall also that the  eigenvalues of $\widetilde{Q}^{1/2} \Sigma \widetilde{Q}^{1/2}$ are the same as those of
	$ \widetilde{Q}^{-1/2}\widetilde{Q}^{1/2} \Sigma \widetilde{Q}^{1/2} \tilde{Q}^{1/2}$ since this is a similarity transform. These observations lead to \eqref{eq:uqf-eigen}.
\end{proof}

\begin{proof}[of Proposition \ref{prop:uqf-bound-pf}]
	
	By assumption $\pi(\theta)\equiv \Nor(\theta; \mu,Q^{-1})$; let $\mu,Q$ be organized according to $\mathcal{U},\mathcal{C}$, hence, e.g., $Q_{\mathcal{UU}}$ denotes the block of the matrix corresponding to the $\theta_k$ for $k\in \mathcal{U}$. Since $\pi(\theta)$ is Gaussian then both $\pi(\theta_{\mathcal{C}} \mid  \theta_{\mathcal{U}})$ and $\pi(\theta_{\mathcal{U}})$ are Gaussian, the precision of the latter is given by $Q_{\mathcal{UU}} - Q_{\mathcal{UC}}Q_{\mathcal{CC}}^{-1} Q_{\mathcal{CU}}$. 
	From the proof of Proposition \ref{prop:PFCAVI} we recall that  $q^*(\theta_{\mathcal{C}} \mid  \theta_{\mathcal{U}}) = \pi(\theta_{\mathcal{C}} \mid \theta_{\mathcal{U}})$, it is unique and  is obtained after a single update, and that the coordinate ascent variational inference iterates for $q(\theta_{k})$, $k \in \mathcal{U}$, minimize the KL divergence between the marginal $\pi(\theta_{\mathcal{U}})$ and its factorized approximation. It is then known (see, e.g., Chapter 10 of \citealt{bishop}) that  for $\pi(\theta_{\mathcal{U}})$ Gaussian, $\{q^*(\theta_{k})\}$ is unique and is Gaussian with mean $\mu_k$ and precision the $k$'th diagonal block of $Q_{\mathcal{UU}} - Q_{\mathcal{UC}}Q_{\mathcal{CC}}^{-1} Q_{\mathcal{CU}}$; let $\widetilde{Q_{\mathcal{UU}}}$ denote the block-diagonal matrix that contains these precision matrices. Therefore, $q^*(\theta_{\mathcal{C}},\theta_{\mathcal{U}})$ is unique, it is Gaussian with mean $\mu$ and precision matrix with the same blocks as $Q$ with the exception of that $\widetilde{Q_{\mathcal{UU}}}$ replaces $Q_{\mathcal{UU}}$.
	
	By appealing to Proposition \ref{prop:uqf} and its proof, we get that $UQF(q^*\Vert \pi)$ is the smallest eigenvalue of the product of the precision under $\pi$ and the covariance under $q^*$. A careful calculation using both of the two standard expressions for the inverse of a block matrix (the one that inverts the top left and the one that inverts the bottom right) shows  that the product is a block lower triangular matrix with diagonal elements the identity matrix for the $\mathcal{CC}$-block and
	$$
	(Q_{\mathcal{UU}} - Q_{\mathcal{UC}}Q_{\mathcal{CC}}^{-1} Q_{\mathcal{CU}})
	(\widetilde{Q_{\mathcal{UU}}})^{-1},
	%(\widetilde{Q_{\mathcal{UU}}} - Q_{\mathcal{UC}}Q_{\mathcal{CC}}^{-1} Q_{\mathcal{CU}})^{-1},
	$$
	for the $\mathcal{UU}$-block. We can recognize the second term in the product as $\mathrm{cov}_{q^*}(\theta_{\mathcal{U}})$, as obtained by block-wise inversion. 
	%However, we know from the factorization and the optimality of mean-field for Gaussian that this is simply $\widetilde{Q_{\mathcal{UU}}}^{-1}$ since the covariance under $q^*$ is block-diagonal.
	Therefore, 
	the $\mathcal{UU}$-block can be expressed generically as $P\widetilde{P}^{-1}$ for $P$ positive definite organized in $K$ blocks and $\widetilde{P}$ block-diagonal containing the diagonal blocks of $P$. Therefore, the matrix that determines $UQF$ will have the eigenvalue 1 with multiplicity $|\mathcal{C}|$ and the rest of its eigenvalues are those of $P\widetilde{P}^{-1}$. We now show that the minimum eigenvalue of such a matrix is less than 1. 
	
	First, note that due to a similarity transformation the eigenvalues of $P\widetilde{P}^{-1}$ are the same as those of $\widetilde{P}^{-1/2}P\widetilde{P}^{-1/2}$. This is a matrix  with  $|\mathcal{U}|$ diagonal blocks that are identity matrices. Consider a vector $w$ with constant entries on the 1st block and 0 everywhere else, normalized to have norm 1. Then
	$$
	\lambda_{min}(\widetilde{P}^{-1/2}P\widetilde{P}^{-1/2})= \inf_{v:v^T v=1} v^T \widetilde{P}^{-1/2}P\widetilde{P}^{-1/2}  v \leq w^T \widetilde{P}^{-1/2}P\widetilde{P}^{-1/2}  w = w^T w = 1\,.
	$$
	
\end{proof}

See Section \ref{sec:rs-cavi} in the Supplement for details on the random scan version of coordinate ascent we work with here. 

\begin{proof}[of Theorem \ref{th:rs-pf-cavi-conv}]

	For given $t$, let $s=|\mathcal{U}| t$, and let $q_{s}(\theta) =   \pi(\theta_{\mathcal{C}}|\theta_{\mathcal{U}}) q_{s}(\theta_{\mathcal{U}})$   be the density returned after $|\mathcal{U}| t$ iterations of random scan coordinate ascent variational inference, where $q_{s}(\theta_{\mathcal{U}}) =  \left[\prod_{k \in \mathcal{U}} q_{s}(\theta_{k})\right]$. Denoting by $q^*(\theta)$ and $q^*(\theta_{\mathcal{U}})$ the corresponding optimal at convergence, the proof of Proposition \ref{prop:uqf-bound-pf} shows that  $UQF(q^*(\theta) \Vert \pi(\theta)) = UQF(q^*(\theta_{\mathcal{U}}) \Vert \pi(\theta_{\mathcal{U}}))$. It also follows easily that $KL(q_{s}(\theta) \vert \pi(\theta)) = KL(q_{s}(\theta_{\mathcal{U}}) \vert \pi(\theta_{\mathcal{U}}))$ with analogous equality for the optimized densities. Therefore,  for the rest of the proof it suffices to focus on the algorithm that approximates $\pi(\theta_{\mathcal{U}})$ by the mean-field $\prod_{k\in \mathcal{U}} q(\theta_k)$ using a random scan coordinate ascent variational inference.
	
	Given the assumption that $\pi(\theta)$ is Gaussian, we have that $\pi(\theta_{\mathcal{U}})$ is Gaussian and let $\mu,Q$ denote its mean and precision, and let $b = Q \mu$.  Then, it is well-known (see, e.g., Chapter 10 of \citealt{bishop}) that when updating the $k$'th term in the approximation, the coordinate-wise optimum is    $ \Nor(m_{k}, (Q_{k,k})^{-1})$, where the mean depends on the means of all other densities via the relation
	\[
	Q_{k,k} m_{k} = Q_{k,k} m_{k} - \sum_{\ell} Q_{k,\ell} m_{\ell} + b_{k}\,.
	\]
	The following basic considerations allow us to simplify the study of convergence of coordinate ascent variational inference in this context. First, note that convergence of coordinate ascent variational inference is precisely that of the convergence of the means $m^{(k)}$ to their stationary point, since  the corresponding precision matrices converge after one iteration to their limit. Second, the convergence of the linear dynamical system given above does not depend on $b$ and remains the same if we do the coordinate-wise reparameterisation $m_{k} \to Q_{k,k}^{1/2}  m_{k}$. When taking also into account the first expression for $UQF$ in Proposition \ref{prop:uqf} we conclude that without loss of generality in the rest of the proof we can take  $\mu = 0$ and that $Q^{(k,k)} = I$ for each block $k$. Notice then that
	\[
	UQF(q^* \Vert \pi) = \lambda_{min}(Q)<1,
	\]
	where the upper bound is due to Proposition \ref{prop:uqf-bound-pf}.  
	Letting $m_t$ be the vector of means after $s$ iterations, a direct calculation yields that
	\[
	V(q^* \Vert \pi) - V(\pi \Vert q_s) = {1\over 2} m_s^T Q m_s\,.
	\]

	Under the above considerations, we can now focus on obtaining convergence bounds for  the random scan coordinate descent for minimizing $V(m) = m^T Q m/2$, with $m$ organized in $|\mathcal{U}|$ blocks and $Q$ having identity diagonal  blocks, with $m_t$ denoting the vector after $t$ random coordinate minimizations.

	Let $\mathcal{F}_s$ be the filtration generated by $I_1,\ldots,I_s$. Then, a careful calculation obtains the following equality, whereas the inequality is a basic eigenvalue bound:
	\[
	\mathbb{E}[V(m_{s+1}) \mid \mathcal{F}_s] = V(m_s) - {1\over |\mathcal{U}|} {1\over 2} m_s^T Q^2 m_s \leq  \left (1 -{1 \over |\mathcal{U}|}\lambda_{min}(Q)\right ) V(m_s)\,,
	\]
	from which we directly obtain 
	\[
	\mathbb{E}[V(m_{|\mathcal{U}|t}]  \leq \left(1 - {\lambda_{min}(Q) \over |\mathcal{U}|}\right)^{|\mathcal{U}|t} V(m_0)\,.
	\]
	
	For the lower bound, note first that
	\[
	\mathbb{E}[m_{s+1} \mid \mathcal{F}_s] = (I- Q/|\mathcal{U}|) m_s\,,
	\]
	an implication of which is that if $m_0$ is an eigenvector of $Q$ with eigenvalue $\lambda$ then $\mathbb{E}[m_{s+1}] = (1-\lambda/|\mathcal{U}|)^{s+1} m_0$. Since for positive semi-definite $Q$, $V(m)$ is a convex function, we can apply Jensen's inequality to obtain:
	\[
	\mathbb{E}[V(m_{|\mathcal{U}|t})] \geq \mathbb{E}[m_{|\mathcal{U}|t}]^T Q  \mathbb{E}[m_{|\mathcal{U}|t}] = (1-\lambda/|\mathcal{U}|)^{2|\mathcal{U}|t} V(m_0)\,.
	\]
	The proof is concluded by taking $m_0$ to be the eigenvector that corresponds to $\lambda_{min}(Q)$.
\end{proof}

\subsection{Proofs of Section \ref{sec:theory-GLMM}}

Throughout these proofs we use the following definitions. We define a  block matrix $N$ that contains the weighted counts $n_{g,h}^{(k,\ell)}$ defined in \eqref{eq:counts}, organized according to the numbering of the factors and their levels. For $Q$ and $\widetilde{Q}$ the precision matrices under the target and the variational approximation, we define
\begin{equation}
\label{eq:q-bar}
\bar{Q} = \widetilde{Q}^{-1/2} Q\widetilde{Q}^{-1/2}.
\end{equation}
When the target is Gaussian, $\mathcal{C} = \emptyset$ implies that $\widetilde{Q}$ coincides with $Q$ on diagonal blocks and it is zero elsewhere, which in turn implies that $\bar{Q}$ is a matrix whose diagonal blocks  are identity matrices.  
The following generic Lemma is used to prove Theorem \ref{th:ff}.

\begin{lemma}
	\label{lem:uqf-bound}
	Consider $Q,\widetilde{Q},\bar{Q}$ as in \eqref{eq:q-bar} and vectors $v_k,v_\ell$ for $k\neq \ell$ such that $v_k$ has only non-zero elements in the $k$-th block, and $v_\ell$ accordingly. Then:
	\begin{equation*}
	UQF(q \Vert \pi) \leq 1 - {v_k^T \bar{Q} v_\ell \over \|v_k\| \|v_\ell\|}\,.
	\end{equation*}
\end{lemma}

\begin{proof}
	For $u_k = v_k/\|v_k\|^2$ and $u_\ell$ defined accordingly, note that they are orthogonal and normalized to 1. Consider the $2\times 2$ matrix
	\begin{equation*}
	C = \left( \begin{array}{c} u_k^T \\ u_\ell^T 
	\end{array}\right) \bar{Q} (u_k \quad u_\ell) \,.
	\end{equation*}
	By construction $C_{ii} = 1$ for $i=1,2$ and $C$ inherits positive definitess from $\bar{Q}$. Then:
	\begin{align*}
	UQF(q \Vert \pi) & = \lambda_{min}(\bar{Q}) \leq \min_{\gamma_1,\gamma_2} {(\gamma_1 u_k  + \gamma_2 u_\ell )^T \bar{Q} (\gamma_1 u_k + \gamma_2 u_\ell )  \over \gamma_1^2+\gamma_2^2}   \\
	& = \min_{\gamma} {\gamma^T C \gamma \over \|\gamma\|^2} = \lambda_{min}(C) = 1-C_{12}\,.
	\end{align*}
	Replacing by $v_k,v_\ell$ completes the proof.
\end{proof}

\begin{proof}[ of Theorem \ref{th:ff}]
	Directly from Proposition \ref{prop:pi} for $\mathcal{C}=\emptyset$, we obtain that the precision of $\pi(\theta)$, denoted by $Q$ below, is 
	\begin{align}
	Q_{k,k} & = T_k I  + \mathrm{diag}[(n_{g}^{(k)})_{g=1:G_k}], \quad \quad Q_{k,\ell} = N_{k,\ell}, \quad k,\ell >0, k\neq \ell \label{eq:Q_def_1}\\
	\quad Q_{0,0} & = \delta, \quad Q_{0,k} = (n_g^{(k)})_{g=1:G_k}, k>0\,,
	\label{eq:Q_def_2}
	\end{align}
	where $\delta=\sum_{g=1}^{G_k} n_{g}^{(k)}= %\sum_{i=1}^n D^2_{ii}=
	n\bar{D}$. 
	It is well known (see also proof of Theorem \ref{th:rs-pf-cavi-conv}) that for such target the optimized approximation of the variational problem in \eqref{eq:fix-phi} is Gaussian with precision $\widetilde{Q}$ given by the diagonal blocks of $Q$ above. Then, for $\bar{Q}$ as in \eqref{eq:q-bar}, a direct calculation yields that
	\begin{align*}
	\bar{Q}_{k,k} & = I, k>0, \quad   \bar{Q}_{0,0} = 1, \quad (\bar{Q}_{k,\ell})_{g,h} = b_{g}^{(k)} b_h^{(\ell)}  n^{(k,\ell)}_{g,h},\quad (\bar{Q}_{0,k})_g = {b_g^{(k)}\over \delta^{1/2}}  n_g^{(k)} , k,\ell>0\,, 
	\end{align*}
	for
	\[
	b_g^{(k)} = \left(1 \over T_k + n^{(k)}_g\right)^{1/2}\,,
	\]
	and where $(\bar{Q}_{k,\ell})_{g,h}$ denotes the $(g,h)$ element of the block $\bar{Q}_{k,\ell}$. 
	Consider now $v_{k},k>0$ a vector with 0's except for positions on the $k$'th block, where the $g$'th position in the $k$th block takes the value $1/b_g^{(k)}$, and $v_{0}$ which is non-zero only the 0'th position and takes the value $\delta^{1/2}$. Then note that:
	\begin{align*}
	v_{0} \bar{Q} v_{k} & = \sum_{g=1}^{G_k} n_g^{(k)}, \quad \|v^{(0)}\| = \left(\sum_i \delta_i\right)^{1/2}, \quad \|v^{(k)}\| = \left(G_k T_k + \sum_{g=1}^{G_k} n_g^{(k)}\right)^{1/2}\,.   
	\end{align*}
	Appealing to Lemma \ref{lem:uqf-bound} we obtain the result. 
\end{proof}

We first state a preliminary lemma.
\begin{lemma}\label{lemma:2_by_2_uqf}
	Let  $M=(M_{k,\ell})_{k,\ell=1,2}$ be a symmetric positive definite $2\times 2$ block matrix, and 
	\begin{equation*}
	C = 
	\left( \begin{array}{cc} M_{1,1}&0\\ 0 & M_{2,2}\end{array}\right) ^{-1/2}
	\left( \begin{array}{cc} M_{1,1} &M_{1,2} \\ M_{2,1} & M_{2,2}\end{array}\right) 
	\left( \begin{array}{cc} M_{1,1}&0\\ 0 & M_{2,2}\end{array}\right) ^{-1/2}
	\,,
	\end{equation*}
	Then, the minimum modulus eigenvalue of $C $ is
	$$\lambda_{min}(C)=1-\rho(M_{1,1}^{-1} M_{1,2} M_{2,2}^{-1}M_{2,2})^{1/2}$$
	where $\rho(A)$ denotes the largest modulus eigenvalue of a matrix $A$. 
\end{lemma}
\begin{proof}
	Let $d_1$ and $d_2$ be the dimensionality of $M_{1,1}$ and $M_{2,2}$, respectively.
	Assume $d_2\geq d_1$ without loss of generality. 
	We have 
	$$
	C = I_{d_1+d_2}+\tilde{C}\,,
	\qquad
	\tilde{C}=  \left( \begin{array}{cc} 0_{d_1}&a\\ a^T & 0_{d_2}\end{array}\right)
	$$
	where $a=M_{1,1}^{-1/2}M_{1,2}M_{2,2}^{-1/2}$, $I_d$ denote a $d\times d$ identity matrix and $0_d$ denotes a $d\times d$ matrix of zeros.
	It follows that the spectrum of $C$ coincides with the one of $\tilde{C}$ traslated by $+1$.
	The matrix $\tilde{C}$ can be interpreted as the adjacency matrix of a bipartite graph and it can be easily shown that its spectrum coincides with the union of $1$, with multiplicity $d_2-d_1$, the singular values of $A$ and the singular values of $A$ with negative sign, see, e.g., \cite{brito2022spectral}. %\gzcom{(HERE WE MAY NEED SOME MORE ASSUMPTIONS, NEED TO CHECK!)} 
	It follows that $\rho(C)=1-\rho(aa^T)^{1/2}$, where we also used that the singular values of $A$ have modulus less than $1$, which follows from $a=M_{1,1}^{-1/2}M_{1,2}M_{2,2}^{-1/2}$ and the positive definiteness of $Q$.
	To conclude, note that 
	$$\rho(aa^T)=
	\rho(M_{1,1}^{-1/2}M_{1,2}M_{2,2}^{-1}M_{1,2}M_{2,2}^{-1/2})
	=
	\rho(M_{1,1}^{-1}M_{1,2}M_{2,2}^{-1}M_{1,2})\,.$$
\end{proof}

\begin{proof}[of Theorem \ref{th:pf}]
	By Proposition \ref{prop:pi} $\pi(\theta)$ is Gaussian. Thus we can apply the proof of Proposition \ref{prop:uqf-bound-pf} to deduce $UQF(q^*(\theta) \Vert \pi(\theta)) = UQF(q^*(\theta_{\mathcal{U}}) \Vert \pi(\theta_{\mathcal{U}}))$.
	By Proposition \ref{prop:uqf}
	\begin{equation*}
	UQF(q^*(\theta_{\mathcal{U}}) \Vert \pi(\theta_{\mathcal{U}})) =
	\lambda_{min}(\mathrm{cov}_{q^*}(\theta_{\mathcal{U}})^{1/2} \mathrm{cov}_\pi(\theta_{\mathcal{U}})^{-1} \mathrm{cov}_{q^*}(\theta_{\mathcal{U}})^{1/2} )\,.
	\end{equation*}
	Denoting the precision matrix of $\pi(\theta_{\mathcal{U}})$ as 
	a $2\times 2$ block matrix 
	$$M=\mathrm{cov}_\pi(\theta_{\mathcal{U}})^{-1}=\left( \begin{array}{cc} M_{1,1} &M_{1,2} \\ M_{2,1} & M_{2,2}\end{array}\right) 
	$$  and recalling that $\mathrm{cov}_{q^*}(\theta_{\mathcal{U}})^{-1}$ is equal to $M$ on the diagonal blocks and zero elsewhere, Lemma \ref{lemma:2_by_2_uqf} implies
	\begin{equation*}
	UQF(q^*(\theta_{\mathcal{U}}) \Vert \pi(\theta_{\mathcal{U}})) =
	1-\rho(M_{1,1}^{-1} M_{1,2} M_{2,2}^{-1}M_{2,2})^{1/2}\,.
	\end{equation*}
	By $\mathcal{C}=\{0\}$, $M=\mathrm{cov}_\pi(\theta_{\mathcal{U}})^{-1}$ and standard properties of Gaussian distributions, we have 
	$$M=Q_{\mathcal{UU}} - Q_{\mathcal{UC}}Q_{\mathcal{CC}}^{-1} Q_{\mathcal{CU}}\,,$$ where $Q_{\mathcal{UU}}$, $Q_{\mathcal{UC}}$, $Q_{\mathcal{CC}}$ and $Q_{\mathcal{CU}}$ denote blocks of the precision matrix $Q=\mathrm{cov}_\pi(\theta)^{-1}$ defined in \eqref{eq:Q_def_1}-\eqref{eq:Q_def_2}.
	Exploiting Assumption \ref{ass:balance} it can be shown that the entries of $M$ satisfy
	\begin{align*}
	M_{k,\ell}
	&=
	\frac{\delta}{G_k}\left(I_k  -G_{k}^{-1}H_{k}\right)K_{k\ell}
	&k\neq \ell\\
	M_{k,k}
	&=
	\left(T_k+G_k^{-1}\delta\right)I_k  -G_k^{-2}\delta H_{k}
	&
	\end{align*}
	where $I_k$ is the $G_k\times G_k$ identity matrix, $H_k$ is the $G_k\times G_k$ matrix of ones and $K_{k\ell}=G_k\delta^{-1}N_{k,\ell}$ are the stochastic matrices defined in the theorem statement.
	
	Noting that $M_{k,k}^{-1}=\left(T_k+G_k^{-1}\delta\right)^{-1}I_k  +c_kH_{k}$ for some $c_k\in\mathbb{R}$ and that $H_{k}\left(I_k  -G_{k}^{-1}H_{k}\right)=O_{k}$, where $O_k$ denotes the $G_k\times G_k$ matrix of zeros, we obtain
	\begin{align*}
	M_{k,k}^{-1}M_{k,\ell}
	&=
	\frac{\delta}{G_kT_k+\delta}\left(I_k  -G_{k}^{-1}H_{k}\right)K_{k\ell}
	&k\neq \ell
	\end{align*}
	and thus
	\begin{align*}
	M_{1,1}^{-1}M_{1,2}
	M_{2,2}^{-1}M_{2,1}
	&=
	\left(\prod_{k=1}^2{\delta \over G_k T_k + \delta }\right)
	\left(I_1  -G_{1}^{-1}H_{1}\right)K_{12}
	\left(I_2  -G_{2}^{-1}H_{2}\right)K_{21}\,.
	\end{align*}
	The statement follows noting that
	\begin{align}\label{eq:sle}
	\rho(
	\left(I_1  -G_{1}^{-1}H_{1}\right)K_{12}
	\left(I_2  -G_{2}^{-1}H_{2}\right)K_{21}
	)
	=
	\rho(
	K_{12}K_{21}\left(I_1  -G_{1}^{-1}H_{1}\right)
	)
	=
	\lambda_2(K_{12}K_{21})
	\,,
	\end{align}
	where $\lambda_2(A)$ denotes the second largest modulus eigenvalue of a matrix $A$.
	The 
	last equality in \eqref{eq:sle} follows from the fact that $K_{12}K_{21}$ is a stochastic matrix, and thus a largest modulus eigenvector for it is given by the vector of ones, which vanishes upon multiplication with the matrix $I_1  -G_{1}^{-1}H_{1}$.
	
\end{proof}

\begin{proof}[of Theorem \ref{th:rg}]
	By Theorem \ref{th:pf}
	$UQF(q^* \Vert \pi) \geq 1 - \lambda_{aux}^{1/2}$.
	%1 - \frac{\sqrt{G_1}+\sqrt{G_2}}{\sqrt{n}}-\epsilon
	Let $d_k=n/G_k$ for $k=1,2$.
	By Theorem 3.2 of \cite{brito2022spectral}, for any $\epsilon' > 0$, we have
	$$\lambda_{aux}\leq \frac{(d_{1}-1)^{1/2}+(d_{2}-1)^{1/2}+\epsilon'}{(d_1d_2)^{1/2}}\leq
	d_{1}^{-1/2}+d_{2}^{-1/2}
	+
	\epsilon'
	$$
	almost surely as $n \to \infty$ for fixed $d_{1}$ and $d_{2}$ (larger or equal than $3$).\footnote{The condition $d_{1}, d_{2}\geq 3$ is assumed at the beginning of \cite{brito2022spectral} and not explicitly recalled in Theorem 3.2 therein.}
\end{proof}

\section{Random-scan coordinate ascent}
\label{sec:rs-cavi}

Theorem \ref{th:rs-pf-cavi-conv} analyzes the random scan version of partially factorized coordinate ascent and for completeness we provide here its description. We focus on the variational problem defined in Assumption \ref{ass:sub-cavi-gauss-pi}. %\eqref{eq:fix-phi}. 
For this problem, as obtained in the proof of Proposition \ref{prop:PFCAVI}, $q^*(\theta_{\mathcal{C}} | \theta_{\mathcal{U}}) = \pi(\theta_{\mathcal{C}} | \theta_{\mathcal{U}})$ and is found after a single update. Therefore, the algorithm updates this term only once and randomly scan through the other terms to perform updates.
In order to match computational costs between $T$ iterations of random and deterministic scan versions, Algorithm \ref{alg:rs-pf-cavi} performs $|\mathcal{U}|T$ random-scan updates instead of $T$. Equivalently, at each random-scan update in Algorithm \ref{alg:rs-pf-cavi} we increase time by $1/K$. 
\begin{algorithm}[htbp]
	\KwData{target distribution $\pi(\theta)$; sets $\mathcal{C}\subseteq \{0,\ldots, K\}$  and $\mathcal{U} = \{0, \ldots, K\} \setminus \mathcal{C}$; Set $t\gets 0$}
	\While{$t<T$}
	{
		Pick $k$ uniformly at random from $\mathcal{U}$\;
		Set $q_{t+1/K}(\theta_{k'})=q_{t}(\theta_{k'}) $ for  $k'\in{\mathcal{U}}\backslash k$\;
		Set
		$
		q_{t+1/K}(\theta_k) \propto 
		\exp\left \{\mathbb{E}_{q_{t}(\theta_{\mathcal{U}\backslash k})}[\log 
		\pi(\theta_{\mathcal{U}})] \right\}
		$\;
		Set $t\gets t+1/K$\;
		%   \EndFor
	}
	\KwResult{$q_T (\theta) 
		=
		\pi(\theta_{\mathcal{C}} \mid \theta_{\mathcal{U}})
		\left[\prod_{k \in \mathcal{U}} q_{T}(\theta_{k})\right]$}
	\caption{Random-scan partially factorized coordinate ascent targeting $\pi$
		\label{alg:rs-pf-cavi}}
\end{algorithm}

\section{Additional numerical experiments}

\subsection{Details on voter turnout models}

First, we provide some additional details on the nine models we consider for voter turnout. The following table summarizes the main characteristics of the models.

\begin{table}[!htbp]
	%\spacingset{1}
	\caption{Description of models}\label{tab:gg_models}
	\begin{center}
		\begin{tabular}{lp{8.5cm}rr}
			Model & Description & $\sum_{k\in\mathcal{C}}G_kD_k$ & $\sum_{k\in\mathcal{U}}G_kD_k$ \\
			Model 1 & Additive model (age, income, state, ethnicity) &   6 &   64\\
Model 2 & Adds random slopes for (continuous) income on age, state, and ethnicity &   6 &  123\\
Model 3 & Adds two-way interactions between age, income, and ethnicity &  27 &  158\\
Model 4 & Adds two-way interactions between state and demographic variables & 129 &  719\\
Model 5 & Adds region (e.g., ``South'') and interactions between region and demographics & 139 &  784\\
Model 6 & Adds three-way interaction between age, income, and ethnicity & 139 &  864\\
Model 7 & Adds three-way interaction between state, ethnicity, and income & 139 & 1884\\
Model 8 & Adds three-way interaction between state, ethnicity, and age & 139 & 2700\\
Model 9 & Adds three-way interaction between state, income, and age & 139 & 3720\\

		\end{tabular}
	\end{center}
\end{table}

\subsection{Comparisons Against Alternative Variational Methods}

We replicate the key results in our manuscript on the UQF (e.g., Figures~\ref{fig:uqf_sim} and Figures~\ref{fig:uqf_gg}) and on the accuracy of predictions in the Ghitza-Gelman example (Figure~\ref{fig:acc_gg}) using two generic methods for variational inference that are implemented in STAN. We do this using \texttt{cmdstanr} (with \texttt{CmdStan} version 2.34.1). ADVI is estimated using its mean-field variant on default settings. For Pathfinder, we examine two versions; first, one (``Pathfinder'') that uses mostly the default settings but adjusted to have 20 paths. Second, we one use with a longer history (``Pathfinder (LH)'') to see if this improves performance. %Initial experiments suggested that 
Since the default setting of using importance resampling resulted in poor performance for the models under consideration, the results presented here do not use this technique.

For the simulations, we draw 100,000 samples from $q(\theta)$ for each variational approximation and use this to estimate the UQF. The longer history Pathfinder uses a history length of 100. In the main manuscript, the UQF for our variational approximations is calculated using the analytical form of $\mathrm{cov}_{q}(\theta)$. To address potential concerns about bias in the estimated UQF for the fixed-form VI methods, Figure~\ref{fig:sim_uqf_supp} reports the results in the main text (``Analytical'') alongside the results using samples from $q(\theta)$ (``Sampled''). The results for fully factorized, partially factorized, and unfactorized variational inference look similar between both estimation methods suggesting that the number of samples used in this section is sufficient to obtain reliable estimates.

\begin{figure}[!ht]
	\caption{Additional UQF results on simulation study}
	\includegraphics[width=0.8\textwidth]{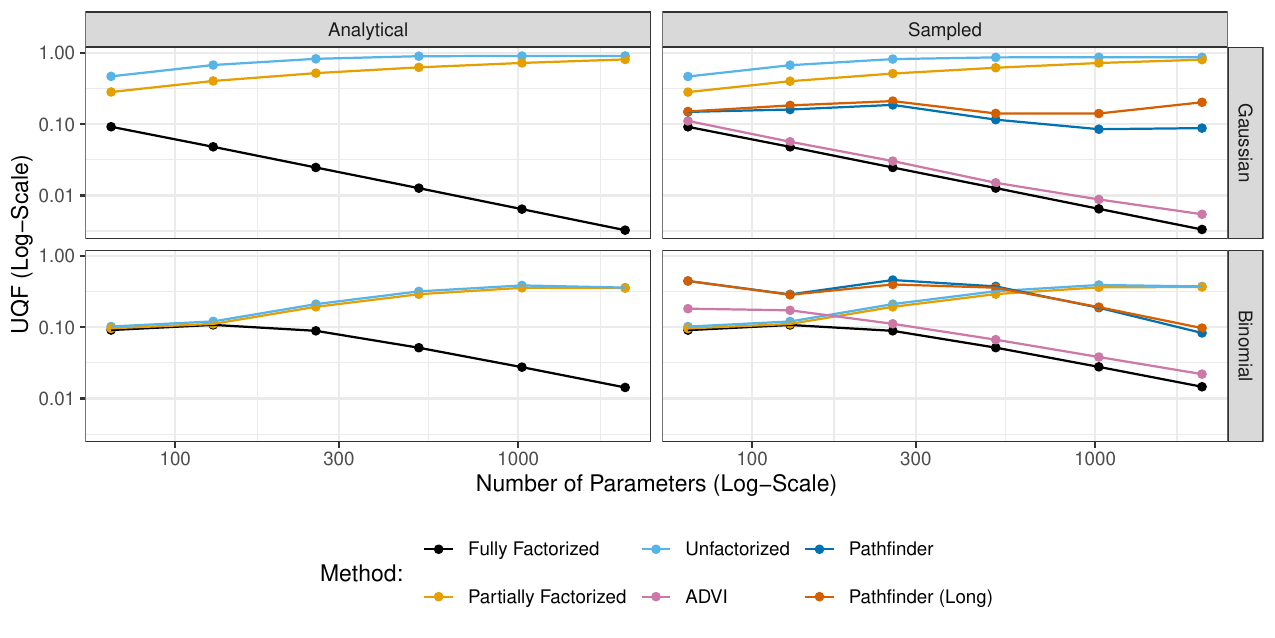}\label{fig:sim_uqf_supp}
\end{figure}

Figure~\ref{fig:sim_uqf_supp} shows that, in the Gaussian case, ADVI behaves similarly to the fully factorized approximation, i.e., the UQF deteriorates as $G_k$ grows. By contrast, both versions of Pathfinder look more similar to the partially-factorized approach: the UQF remains relatively stable with $G_k$, albeit below the values estimated from partially factorized and unfactorized variational inference. In the binomial setting, Pathfinder and ADVI both have better UQF than fully factorized, partially factorized and unfactorized variational inference for smaller models, although for the larger models they perform worse.

\begin{figure}[!ht]
	\caption{Additional UQF results for voter turnout}
	\includegraphics[width=0.8\textwidth]{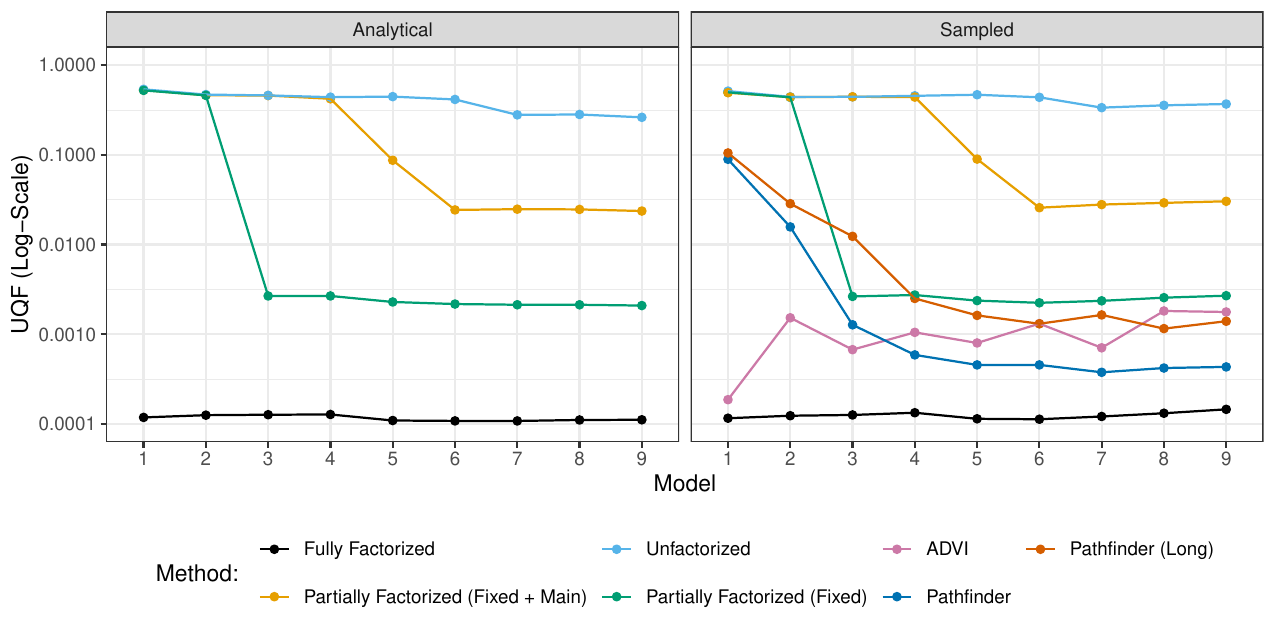}\label{fig:gg_uqf_supp}
\end{figure}

For the Ghitza-Gelman example, we draw 20,000 samples from $q(\theta)$ and use this to calculate the UQF. The longer history Pathfinder uses a history length of 150, slightly larger than the maximum size of the collapsed set on the most complex model (139). Figure~\ref{fig:gg_uqf_supp} shows that partially factorized variational inference (with fixed and main effects in the collapsed set) outperforms Pathfinder (both versions). Both versions of Pathfinder out-perform fully factorized variational inference, while ADVI is typically between the two and somewhat unstable across different models.

\begin{figure}[!ht]
	\caption{Additional accuracy results for voter turnout}	\label{fig:gg_acc_supp}
	\includegraphics[width=0.8\textwidth]{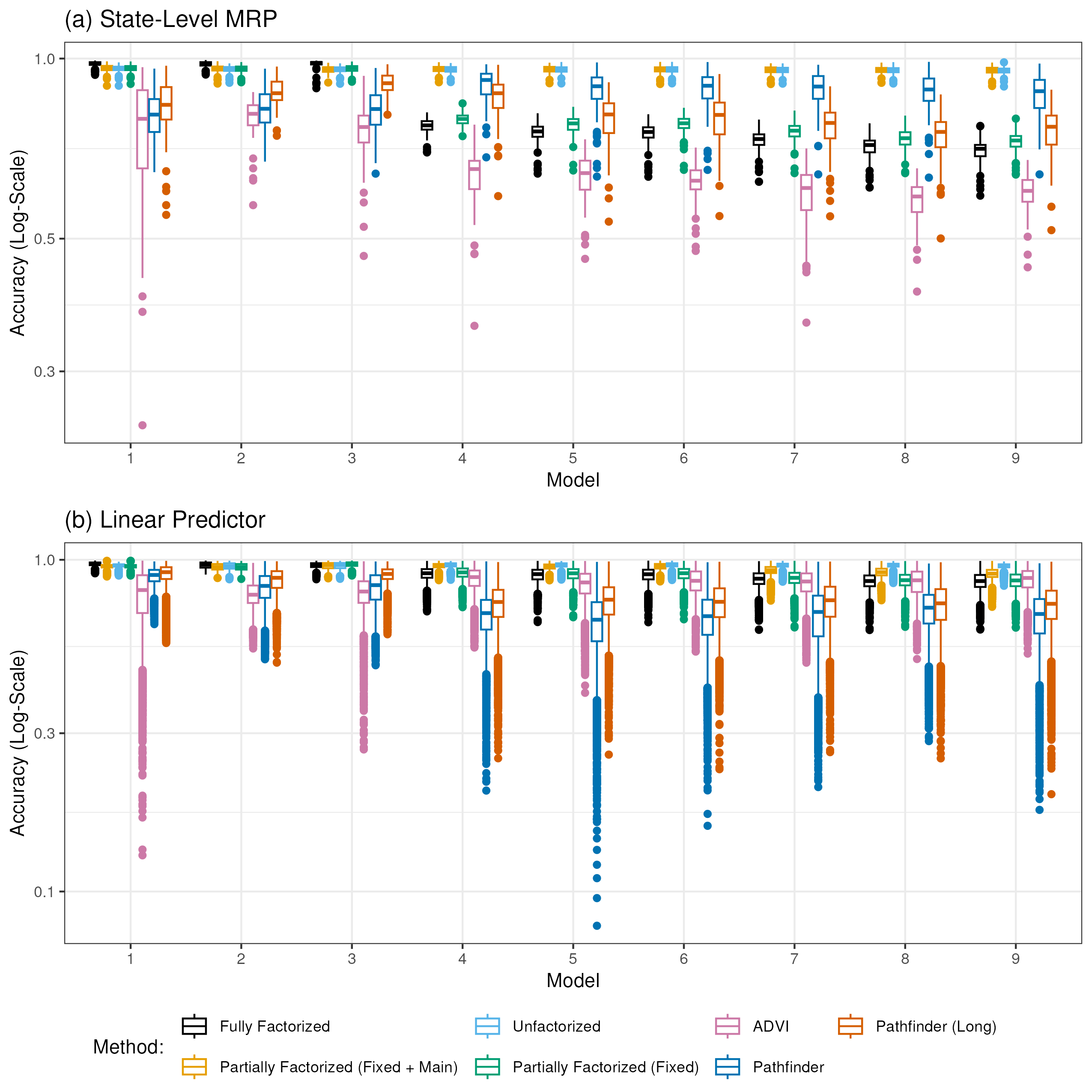}
\end{figure}

Finally, Figure~\ref{fig:gg_acc_supp} shows the accuracy on the post-stratified state-level prediction and 4,080 types of individuals. This shows more clearly the benefit to unfactorized, partially factorized, and fully factorized variational inference since, for the simplest models, all clearly out-perform Pathfinder and ADVI. As the model complexity grows, only unfactorized and partially factorized variational inference show consistently strong performance.

\subsection{Accuracy of Variance Components}

For the simulations, Figure~\ref{fig:varcomp_sim} reports the average accuracy of the variational approximations of the two variance components across all of the simulations. Its results are qualitatively consistent with the results on the UQF reported in the main manuscript, i.e. fully factorized variational inference is worse than partially factorized which is worse than unfactorized. In the Gaussian setting, there is clear evidence of the blessing of dimensionality as the accuracy increases and approaches one as the number of levels in each factor increases. In the binomial setting, this eventually appears asymptotically as the problem grows larger, although the accuracy remains lower than in the Gaussian case. This mirrors the patterns seen in the UQF results in the main text.

\begin{figure}[!htbp]
	\caption{Accuracy for variance components (simulation study)}\label{fig:varcomp_sim}
	\includegraphics[width=0.8\textwidth]{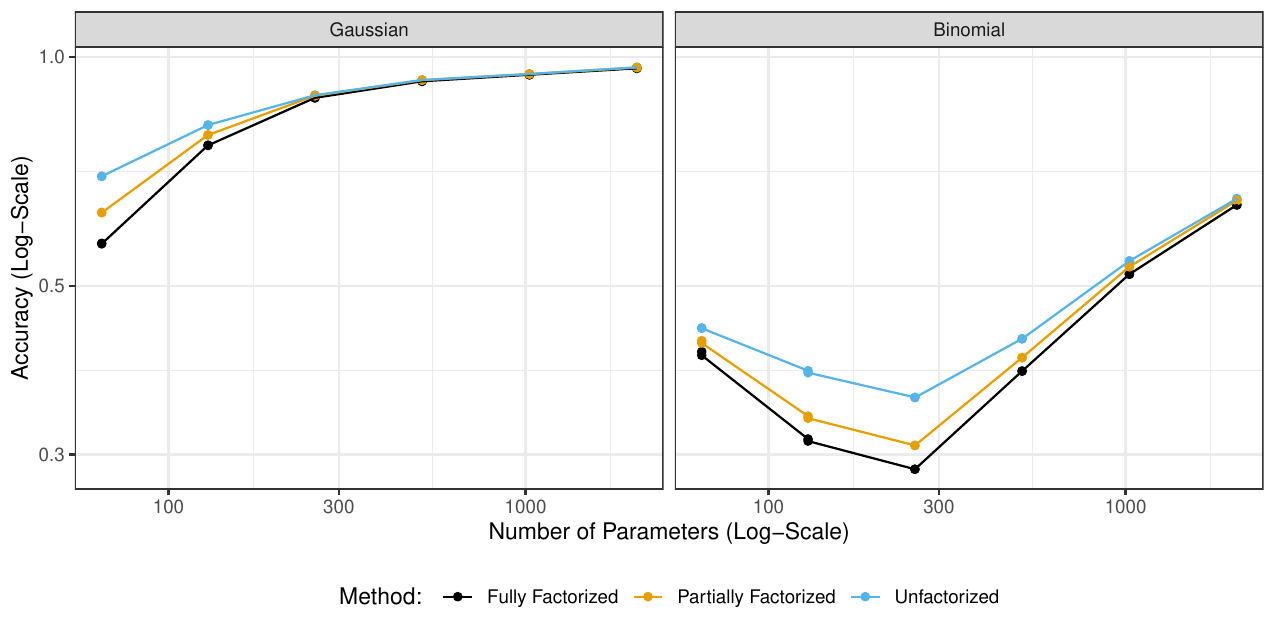}
\end{figure}

\begin{figure}
	\caption{Accuracy for variance components (Ghitza-Gelman)}\label{fig:varcomp_gg}
	\includegraphics[width=0.8\textwidth]{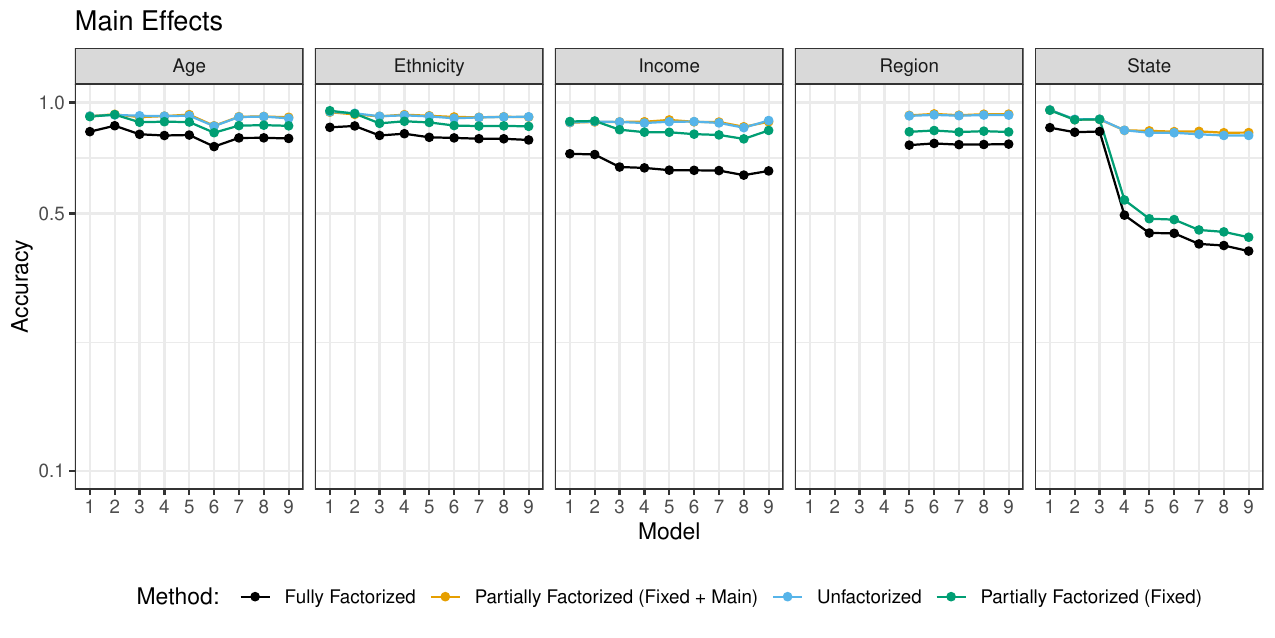}
	\includegraphics[width=0.8\textwidth]{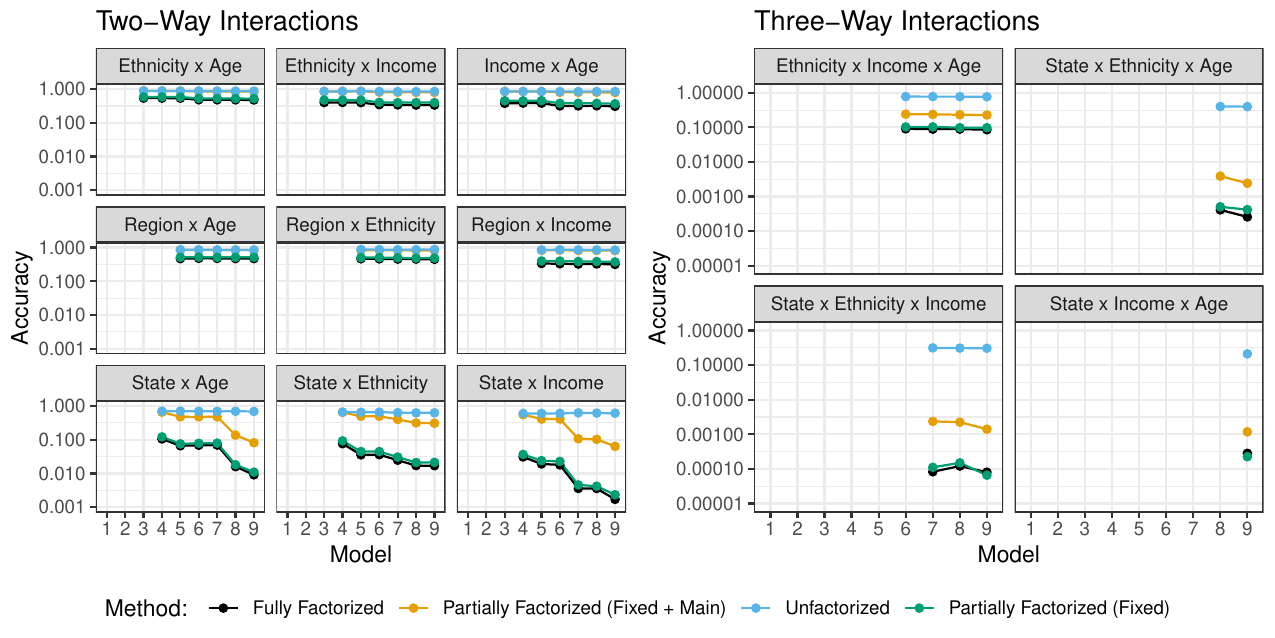}
\end{figure}
For the Ghitza-Gelman example, we calculate the accuracy of each variance component, averaged across both years in the data. For factors with random slopes (i.e., $D_k = 2$), we calculate the accuracy separately for each of the $D_k^2$ components and average them together. Figure~\ref{fig:varcomp_gg} shows results for accuracy on the main effects (i.e., age, ethnicity, income, state, region), two-way interactions (e.g., state-income), and three-way interactions (e.g., age-state-income) separately to allow for the differences to be more clearly displayed.

It suggests that all methods (fully factorized, partially factorized, and unfactorized variational inference) do well for the main demographic effects (age, ethnicity, and income) across all models. However, from Model 3 onwards, there is a sharp decline in the accuracy for fully factorized variational inference and partially factorized variational inference with only fixed effects in $\mathcal{C}$; this corresponds to the model where the UQF for that form of partially factorized variational inference also drops. Looking at the two- and three-way interactions, there are certain factors for which there is a clear gap between partially factorized variational inference (with main effects in $\mathcal{C}$) and unfactorized variational inference and others where performance is similar. However, in all cases, partially factorized variational inference (with main effects in $\mathcal{C}$) performs considerably better than fully factorized variational inference. Broadly speaking, the qualitative patterns of accuracy mirror that shown in the UQF in Figure~\ref{fig:uqf_gg}.

\end{document}